\documentclass{article}

\usepackage{arxiv}

\usepackage[utf8]{inputenc} % allow utf-8 input
\usepackage[T1]{fontenc}    % use 8-bit T1 fonts
\usepackage{hyperref}       % hyperlinks
\usepackage{url}            % simple URL typesetting
\usepackage{booktabs}       % professional-quality tables
\usepackage{amssymb}
\usepackage{amsmath}
\usepackage{nicefrac}       % compact symbols for 1/2, etc.
\usepackage{microtype}      % microtypography
\usepackage{lipsum}		% Can be removed after putting your text content
\usepackage{pdflscape}
\usepackage{booktabs}
\usepackage{graphicx}
\usepackage{subcaption}
\usepackage{longtable}
\usepackage{tabularx}
\usepackage{enumitem}
\usepackage[table]{xcolor}
\usepackage{multirow} % セルを縦に結合するために使用
\usepackage{array}
\usepackage{url}
\usepackage[numbers]{natbib}
\usepackage{doi}

\usepackage{amsthm}

\theoremstyle{definition}
\newtheorem{dfn}{Definition}[section]

\newtheorem{lem}[dfn]{Lemma}
\newtheorem{thm}[dfn]{Theorem}

\title{``Sail Fast, Then Wait'' in First-come, First-served Port Queues: Information Sharing for Sustainable Shipping}

\author{ \href{https://orcid.org/0000-0002-4774-1506}{\includegraphics[scale=0.06]{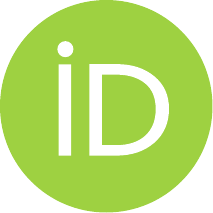}\hspace{1mm}Ayato Kitadai}\thanks{\texttt{a.kitadai@css.t.u-tokyo.ac.jp}}\\
    School of Engineering\\
	The University of Tokyo\\\\
	\And
	Shunta Yoshimura \\
    Department of Engineering\\
	The University of Tokyo\\
    \And
	Takuya Nakashima \\
    School of Engineering\\
	The University of Tokyo\\
    \And
	Noora Torpo \\
    NAPA Ltd\\
    \And
	Rei Miratsu \\
    Nippon Kaiji Kyokai\\
    \And
	Naoki Mizutani \\
    NAPA Japan\\
    \And
	\href{https://orcid.org/0000-0002-6411-8716}{\includegraphics[scale=0.06]{orcid.pdf}\hspace{1mm}Nariaki Nishino} \\
    School of Engineering\\
	The University of Tokyo\\
}

\renewcommand{\shorttitle}{SFTW in FCFS Port Queues: Information Sharing for Sustainable Shipping}

\hypersetup{
pdftitle={``Sail Fast, Then Wait'' in First-come, First-served Port Queues: Information Sharing for Sustainable Shipping},
pdfsubject={q-bio.NC, q-bio.QM},
pdfauthor={Ayato Kitadai},
pdfkeywords={Queueing games, Information sharing, Port congestion, Sustainable shipping, AIS data},
}

\begin{document}
\maketitle

\begin{abstract}
This study develops a novel class of queueing game to explain a common practice in cargo shipping ``Sail Fast, Then Wait'' (SFTW), and demonstrates that resolving information asymmetry among ships can deconcentrate port arrival times.
We formulate a competitive navigating environment as an incomplete information game where players strategically decide their arrival time within heterogeneous feasible sets under First-Come, First-Served port policy.
Our results show that in incomplete information settings, SFTW emerges as the unique symmetric equilibrium.
Conversely, under complete information, the set of equilibria expands, allowing for slower and more environmentally friendly actions without compromising service order.
We further quantitatively evaluate the effect of information enrichment based on empirical data.
Our findings suggest that the prevalence of technologies enabling ships to infer others' private information can effectively reduce SFTW and enable more energy-efficient and environmentally sustainable operations.
\end{abstract}

% keywords can be removed

\section{Introduction}
During the 19th century's ``Great Tea Race,'' clipper ships competed to be the first to dock in London, securing premiums for the earliest arrivals~\cite{MacGregor_1983}. 
This competition cemented a mentality of speed in maritime trade, where sailing fast was not merely desirable but commercially advantageous.
This mindset, though originating in the age of sail, persists today in what is widely recognized as the ``Sail Fast, Then Wait (SFTW)'' phenomenon\footnote{This is also called \textit{Rush to Wait.}}: ships maximize their speed to arrive at ports as early as possible, only to wait at anchor for berthing~\cite{Sung_2022}.

While the context has changed, the structure remains.
The persistence of SFTW in modern maritime logistics is not a remnant of irrational behaviour but rather the product of institutional incentives.
Most notably, many ports allocate berths based on a First-Come, First-Served (FCFS) policy.
Simultaneously, traditional charter party contracts have long stipulated that vessels must proceed with ``utmost dispatch,'' effectively obligating ship operators to prioritize speed regardless of congestion at destination ports~\cite{Alvarez_2010}.
In addition, demurrage clauses incentivize early arrival: as soon as a ship drops anchor, laytime begins, positioning the shipowner to claim compensation for port delays~\cite{Adland_2018}.
In this environment, arriving early is not merely prudent but rational, even if it implies long idle times and inefficiencies.

At the same time, mounting scrutiny of greenhouse gas (GHG) emissions has begun to challenge this logic.
The International Maritime Organization adopted an enhanced GHG strategy in 2023 that sets a net-zero ambition around 2050 with interim checkpoints~\cite{IMO_2023_GHG}.
Contractual and operational instruments are evolving accordingly.
BIMCO’s Just in Time Arrival Clause for voyage charters enables owners and charterers to share information and adjust speed to meet an agreed time of arrival~\cite{BIMCO_JIT_2021}.
Industry initiatives further explore coordinated arrival management; for example, the Blue Visby Consortium has presented its coordinated arrival approach and prototype trials to IMO as an information paper~\cite{BlueVisby_MEPC82}.
These developments reflect a growing recognition that SFTW behaviour, once protected by contract and custom, should give way to more sustainable practices.

Yet despite these promising directions, a fundamental theoretical gap remains.
Specifically, while it is widely assumed that changing contractual terms or sharing arrival information can mitigate SFTW~\cite{Jimenez_2021}, the underlying question of why SFTW emerges as a rational strategy in the first place has not been formally resolved.
Existing studies have modelled berth allocation or optimised ship routing~\cite{Wang_2025, Zhen_2025, Jensen_2025}, but as far as we know, none have captured the strategic decision-making process behind ETA selection in FCFS competitive environments, particularly under incomplete information.

This paper seeks to address that gap.
We present a game-theoretic model of ship arrival under FCFS berth allocation rules, where vessels differ in their earliest possible arrival times and face incomplete information about the earliest feasible arrival times (capabilities) of other vessels.

Methodologically, our model belongs to the class of finite-player FCFS bottleneck games.
Within this class we introduce three features that are particularly suited to competitive port calls.
First, types are earliest feasible arrival times rather than waiting cost parameters, so heterogeneity reflects technological constraints that can be inferred from vessel tracking data such as AIS (Automatic Identification System) data instead of unobservable preferences.
Second, actions are continuous arrival times constrained by these types, and queueing is fully endogenous under FCFS service.
Third, we employ minimal preference assumptions: each vessel strictly prefers earlier service completion to later completion and is indifferent across arrival and waiting time decompositions that yield the same completion time.
All equilibrium results are derived on this continuous action space under these ordinal preferences, which makes our characterisation of SFTW behaviour robust to the exact specification of delay disutility.

Based on this model, we first show that in such an environment, SFTW behaviour arises as the unique symmetric Bayesian Nash equilibrium: all ships rush to arrive as early as possible.
Importantly, this rational outcome emerges not from misaligned preferences or bounded rationality, but from the structure of the game itself.

We then examine the counterfactual scenario of a complete information environment.
We demonstrate that when vessels' capabilities become common knowledge, the set of equilibria expands, allowing players to rationally choose later ETAs without suffering utility loss.
This shift implies that information sharing alone can reduce congestion and enable more environmentally sustainable behaviours.

Understanding SFTW as a rational equilibrium behaviour is essential.
Without this structural understanding, it is difficult to rigorously evaluate the true impact of revised contractual clauses or new coordination mechanisms. 
Simply permitting slow steaming in a contract does not guarantee its adoption.
Instead, only by grasping how incentive structures shape strategic choices can we assess the effectiveness and limits of institutional reforms.
In this light, our model not only explains persistent inefficiencies but also offers a theoretical foundation for designing mechanisms that align private incentives with collective sustainability goals.

The main contributions of this paper are summarized as follows:
\begin{enumerate}
    \item \textbf{Theoretical Contribution:} We formulate a finite-player FCFS bottleneck game in continuous time in which heterogeneity arises from technological constraints, represented by earliest feasible arrival times, while preferences over service completion times remain homogeneous and ordinal.
    Within this framework we characterise the pure strategy Nash equilibria under complete information and prove that, under incomplete information regarding these constraints, SFTW emerges as the unique symmetric Bayesian Nash equilibrium.

    \item \textbf{Managerial Implication:} We demonstrate that simply sharing information about feasibility expands the set of equilibria.
    This reveals that ``Green Navigation'' (slower steaming) can be established as a focal equilibrium without altering the service order, thereby mitigating SFTW within the existing FCFS framework.
    \item \textbf{Empirical Quantification:} By applying our model to AIS data from Port Hedland, we quantify the distribution of the hidden ``slack'' in current arrival sequences. 
    In this highly congested setting, the estimated slack is concentrated in a subset of voyages, implying that utility-preserving slow steaming opportunities are heterogeneous and context-dependent.
    Nevertheless, the analysis provides an empirical benchmark for the magnitude and prevalence of coordination gains that can be unlocked through information transparency within the existing FCFS framework.
\end{enumerate}

The remainder of this paper is organized as follows.
Section~\ref{sec:literature} reviews the related literature.
Section~\ref{sec:model} formulates the game-theoretic model of ship arrival under FCFS competition.
Section~\ref{sec:analysis} presents the equilibrium analysis, deriving the unique symmetric Bayesian Nash equilibrium under incomplete information and the expanded equilibrium set under complete information.
Section~\ref{sec:empirical} applies our theoretical framework to real-world AIS data from Port Hedland, discussing equilibrium selection and quantifying the potential slack available for slow steaming.
Section~\ref{sec:discussion} discusses the implications of our findings for maritime policy and acknowledges limitations.
Finally, Section~\ref{sec:conclusion} concludes the paper.

\section{Literature Review}\label{sec:literature}

Strategic timing of arrivals to a waiting queue has long been studied~\cite{Hassin_2003, Haviv_2021}.
Since Vickrey~\cite{Vickrey_1969}, much of the literature has analysed congestion dynamics using fluid models and characterised equilibrium arrival patterns for a continuum of users~\citep{Arnott_1990, Arnott_1993, Palma_2013, Platz_2017, Wu_2019, Li_2020, Yu_2025}.
Such limits are often effective approximations of discrete phenomena.
However, they can be inappropriate in settings with a modest number of strategic players such as vessels calling at the same port.

To capture these contexts, a stream of papers studied games with a finite number of players.
Most of this work assumes homogeneous agents and focuses on symmetric equilibria~\citep{Rapoport_2004, Otsubo_2008, Rapoport_2010, Breinbjerg_2016, Breinbjerg_2017, Sakuma_2020, Alon_2022, Alon_2023, Breinbjerg_2024}.
While analytically convenient, homogeneity is a poor fit for situations such as maritime operations where ships differ in technical capabilities and voyage conditions.

There is also a literature that incorporates heterogeneity in queueing games and bottleneck models.
These papers typically model preference heterogeneity, mainly heterogeneous waiting costs, and study the resulting equilibria~\citep{Ha_2001, Afeche_2013, Gavirneni_2016, Wan_2017, Silva_2017, Maglaras_2018}.
This modelling is highly important because, in service contexts, considering heterogeneity among user preferences is essential for service providers~\citep{Islier_2024, Tuncalp_2024}.
In contrast, the salient heterogeneity in our context lies in feasible action sets: ships face different earliest feasible arrival times driven by position and condition.

Another longstanding theme is information availability in queueing systems~\citep{Naor_1969, Edelson_1975, Economou_2022}.
Many studies examine how information provided to players, usually about the current queue length, alters strategic behaviour~\citep{Islier_2024, Tuncalp_2024, Hassin_1986, Seale_2005, Stein_2007, Guo_2007, Simhon_2016, Kim_2017, Hu_2018, Hassin_2020, Dimitrakopoulos_2021, Economou_2024}.
By contrast, our analysis centres on information about others' feasible actions (types), namely their earliest feasible arrival times.
To our knowledge, the strategic effects of revealing such feasibility information in FCFS arrival games have received limited attention.

In the specific domain of maritime operations, the SFTW phenomenon is well documented.
However, existing studies typically treat arrival times as exogenous or focus on berth allocation optimisation~\citep{Zhen_2025, Jensen_2025, Zhen_2015}, rarely endogenising the strategic interaction among vessels competing for a bottleneck under information asymmetry.

Taken together, the literature establishes rich insights into strategic arrivals under congestion but leaves a clear gap at the intersection of finite-player FCFS environments, heterogeneity in physical feasibility rather than waiting costs, and information structures regarding these constraints.
Most existing models either assume homogeneous players or focus on preference-based heterogeneity, and they rarely capture the strategic implication of feasibility information in maritime bottlenecks.
Within this existing class of finite-player FCFS arrival games, our model introduces technological heterogeneity in earliest feasible arrival times as private information.
By analysing a Bayesian timing game with a continuum of such types, focusing on symmetric strategies, alongside its complete information counterpart, we provide a microfoundation for the SFTW phenomenon and explore information-based solutions for sustainable shipping.

\section{Model}\label{sec:model}
In this study, we formulate the queueing situation as a simultaneous-move game with finite number of players.
Once players decide the time at which they will reach the destination, they navigate accordingly.
We clearly distinguish between arrival time and service time.
Arrival time corresponds to the players' own decision making but they may not be able to come to shore just after the arrival because of queue.
Since the arrival time, players experience waiting time and then reach to the service time.
Importantly, we assume that no queue exists when the game begins.
This assumption is without loss of generality; as discussed in Section~\ref{subsec:complete_info}, a non-zero initial backlog can be treated as a fixed time constraint without altering the strategic structure of the game.

Let $N=\{1,2,\dots,n\}$ denote the set of vessels going to the same port from different locations, and assume that each vessel has distinct technical characteristics such as maximum speed.
Time is continuous and indexed by $\mathbb{R}$.
Normalizing the port's processing capacity to one unit of workload per unit time, we assume that every vessel carries a workload of $\gamma>0$; equivalently, the port requires $\gamma$ units of time to complete service for a single vessel.
Each vessel $i\in N$ simultaneously chooses a desired arrival time $s_i\in S_i\equiv[t_i,\infty)$, computes a route and speed profile so that its estimated time of arrival (ETA) equals $s_i$, and sails accordingly.
Here, $t_i\in\mathbb{R}_{\ge0}$ is the theoretical lower bound on the arrival times that vessel $i$ can feasibly attain.
Because $t_i$ depends on technical capability, current location, and departure time, it may differ across vessels and is considered as private information.
Without loss of generality, we can assume that $t_1\le t_2\le \dots \le t_n$.
Note that while this structure could formally be represented by a type-dependent utility assigning low payoffs to infeasible choices, we encode these technological constraints directly.
This allows $t_i$ to retain a clear physical interpretation recoverable from AIS data.

At most one vessel can be processed by the port at any time, and service is provided on a first-come, first-served (FCFS) basis.
As a tie-breaking rule, when multiple vessels arrive simultaneously, their service order is determined at random.
Upon arriving at time $s_i\in\mathbb{R}_{>0}$, vessel $i$ must wait until the immediately preceding vessel's service is completed; denote this waiting time by $w_i(s_i,s_{-i})$.
Whenever a backlog exists (i.e., some waiting time is strictly positive), the port operates without idleness once service has begun.

Let $g:\mathbb{N}\to N$ be a bijection mapping service positions to vessels, so $g(k)$ is the vessel in the $k$-th service position, and let $g^{-1}$ be its inverse.
Then the waiting time of vessel $i\in N$ under service order $g$ is described as follows.\footnote{See \ref{app1} for the derivation.}
\begin{equation}
    w_i^g(s_i, s_{-i}) = \max_{1\le k\le g^{-1}(i)}\big\{\,s_{g(k)} + \big(g^{-1}(i) - k\big)\gamma - s_i\,\big\}
    \label{eq:waiting_time}
\end{equation}

Because simultaneous arrivals are tie-broken at random, multiple service orders $g$ may arise.
Let $G(s_i,s_{-i})$ denote the set of all such orders consistent with the strategy profile $(s_i,s_{-i})$.
The ex-ante expected waiting time for player $i\in N$ is then
\begin{equation}
    \mathbb{E}\big[w_i(s_i, s_{-i})\big] = \sum_{g\in G(s_i, s_{-i})} \frac{1}{\lvert G(s_i, s_{-i})\rvert}\, w_i^g(s_i, s_{-i}) .
\end{equation}

Each player wishes to complete service and deliver the cargo as early as is operationally feasible.
Since the service time equals the sum of the chosen arrival time and the realized waiting time, we refer to $s_i + \mathbb{E}\bigl[w_i(s_i, s_{-i})\bigr]$ as the expected service completion time under the profile $(s_1,\dots,s_n)$.
The utility of player $i\in N$ given an arrival-time profile $(s_i)_{i\in N}\in \mathbb{R}_{\ge 0}^n$ is defined via a function $f_i: S_i\times\mathbb{R}\to\mathbb{R}$ as
$u_i(s_i, s_{-i}) = f_i\left( s_i, \mathbb{E}\left[w_i(s_i, s_{-i})\right] \right)$.
Here, we assume that, for arbitrary player $i\in N$, the function $f_i$ satisfies the following characteristics:

\begin{enumerate}
    \item Earlier expected service time is strictly preferred:
    \begin{equation*}
        s_i, s_i' \in S_i,\; w_i, w_i' \in \mathbb{R}_{\ge0},\; s_i + w_i < s_i' + w_i' \Rightarrow f_i(s_i, w_i) > f_i(s_i', w_i'),
    \end{equation*}

    \item Waiting time itself does not matter:
    \begin{equation*}
        s_i, s_i' \in S_i,\; w_i, w_i' \in \mathbb{R}_{\ge 0},\; s_i + w_i = s_i' + w_i' \Rightarrow f_i(s_i, w_i) = f_i(s_i', w_i').
    \end{equation*}
\end{enumerate}

In this study, we consider two kinds of information structure to clarify the effect of information availability to the strategic arrival times.
The first is complete information game where all players know other players' types $t_{-i}$.
The second is incomplete information game where no player know others' types.
% The last is the intermediate case where some players know all other players' types $t_{-i}$ but the others only know their own type.
Under each information environment, each player make decision to maximize their own utility.

\section{Equilibrium Analysis}\label{sec:analysis}
% This study focuses on information structures regarding players' feasible actions.
% By comparing equilibria under these environments, we can reveal the effect of information.

\subsection{Complete Information Game}\label{subsec:complete_info}
Firstly, for the ease of analysis, we consider the complete information game and mathematically find pure strategy Nash equilibria.
As a preparation, we define the expected service order and show one important Lemma.

\begin{dfn}
    Given a strategy profile $(s_1,\dots,s_n)$, we define the expected service order of player $i\in N$ as
    \begin{equation}
        \mathbb{E}\bigl(g_{(s_i, s_{-i})}^{-1}(i)\bigr) := \sum_{g\in G(s_i, s_{-i})} \frac{1}{\lvert G(s_i, s_{-i})\rvert}\, g^{-1}(i).
    \end{equation}
\end{dfn}

\begin{lem}\label{lem:faster_better}
For any strategy profile $(s_1,\dots,s_n)$, if a player can improve her expected service order, she can strictly improve her payoff by that.
\end{lem}

If a player set her ETA earlier to overtake other players, her waiting time may increase.
However, this lemma suggests that the shortened portion of arrival time for overtaking is larger than the extended portion of expected waiting time.
The proof is given in \ref{app2}.

Lemma~\ref{lem:faster_better} excludes the possibility of not deliberately overtaking, which leads fierce competitive environments.
Thus, players will set their ETA so as not to be overtaken by any players whose type is larger than his own.

\begin{lem}\label{lem:no-takeover}
    $S_{n+1} := \{\infty\}$.
    If a strategy profile $(s_1^*, \dots, s_n^*)$ is a Nash equilibrium, $\forall i\in N, s_i^* \in \{t_i\}\cup (S_i\setminus S_{i+1})$.
\end{lem}

Taking the contraposition of Lemma~\ref{lem:faster_better}, we can obtain a requirements for Nash equilibria that no one can improve her expected service order.
Thus, players in Nash equilibria will set their arrival times not to be overtaken by any other players.
Formal proof is given in \ref{app3}.

As a last preparation to find Nash equilibria, we will show one intuitively evident lemma.
This means that, under a given strategy profile, if some players can be switched just depending on random seeds of tie breaking, their ETA must be common.

\begin{lem}\label{lem:same_arrival_time}
    $\forall s\in \prod_{i\in N}S_i, \exists g, g' \in G(s), \exists k\in\mathbb{N}, g(k)\neq g'(k)\Rightarrow s_{g(k)} = s_{g'(k)}$
\end{lem}

\begin{proof}
    We argue by contradiction. 
    Assume $s_{g(k)}<s_{g'(k)}$ without loss of generality.
    Since $g(k)\neq g'(k)$, there exists $j\in\{1,\dots,k\}$ such that $g'^{-1}\!\bigl(g(j)\bigr)\ge k$.
    By definition, for any $g\in G(s)$, $a\le b \Rightarrow s_{g(a)}\le s_{g(b)}$.
    Applying this to $g$ and $g'$ yields
    \begin{itemize}
        \item $s_{g(j)}\le s_{g(k)}$ because $j\le k$,
        \item $s_{g'(k)}\le s_{g'(g'^{-1}(g(j)))}=s_{g(j)}$ because $g'^{-1}(g(j))\ge k$.
    \end{itemize}
    Thus we obtain $s_{g(j)} \le s_{g(k)} < s_{g'(k)} \le s_{g(j)}$, which is a contradiction.
\end{proof}

This and Lemma~\ref{lem:no-takeover} allow us to calculate equilibria by focusing on a service order $g^*$ which corresponds to players' type order: $\forall i\in N, g^*(i) = i$.

While we assumed no initial queue in Section~\ref{sec:model}, the presence of an initial backlog can be rigorously modeled by introducing a virtual player $0$ whose service completion time represents the time the port becomes available.
Letting $t_0$ denote this initial boundary (where $t_0$ represents the backlog clearance time, or $t_0 = t_1$ implies no backlog), we can recursively define the equilibrium strategy sets.
Finally, we obtain the main result in complete information situation.

\begin{thm}\label{thm:nash_equilibria}
    Let $t_0$ be the service completion time of the virtual player $0$ (or $t_0 = t_1$ in the absence of a backlog), and $t_{n+1} = \infty$.
    We can recursively define
    \begin{equation}
        \Theta_i =
        \begin{cases}
            \{t_i\} \cup ([t_i, t_{i+1})\cap [-\infty, \theta_i])\\
            \quad \text{ where }\theta_i = \max\limits_{1\le k\le i-1}\{s_k^* + (i-k)\gamma\} &(\text{if $t_{i-1} < t_i < t_{i+1}$}),\\
            \{t_i\} \quad &(\text{Otherwise}).
        \end{cases}
    \end{equation}
    Strategy profile $(s_1^*,\dots,s_n^*)$ is a pure-strategy Nash equilibrium if and only if $\forall i\in N, s_i^*\in \Theta_i$.
\end{thm}

$\theta_i$ represents when player $i-1$'s service will terminate in Nash equilibria.
If there exists a time interval between before $i-1$'s service terminate and player $i+1$ can invade, any time point in this interval is best response for player $i$ because his service time will be $\theta_i$ and no other player can disturb it.
The decent proof is given in \ref{app4}.

\subsection{Incomplete Information Game}\label{subsec:incomplete-info-game}
Considering real-world situations such as vessel operations, the assumption that all players possess common knowledge of each other's earliest arrival times $t_i$ is exceedingly strong.
Rather, $t_i \in \mathbb{R}_{\ge 0} \equiv T$ is private information and each player makes decisions based on beliefs regarding other players' types $t_{-i}$.
We formulate this situation as a Bayesian game.

Let $T^n$ be endowed with the Borel $\sigma$-algebra $\mathcal{B}$ and let $\lambda$ denote the Lebesgue measure on $(T^n,\mathcal{B})$.
A common prior is a probability measure $P$ on $(T^n,\mathcal{B})$ that is absolutely continuous with respect to $\lambda$ and admits a density $p : T^n \to [0,\infty)$, that is,
\begin{equation*}
    P(B) = \int_B p(t)\,dt \quad \text{for all } B \in \mathcal{B},
\end{equation*}
where $p(t) > 0$ for $\lambda$-almost all $t \in T^n$.
With a slight abuse of notation, we write $p(t_i,t_{-i})$ for the density evaluated at the type profile $(t_i,t_{-i})$, and $dt$ denotes integration with respect to the Lebesgue measure on $T^n$.

Bayesian Nash equilibrium is defined as a strategy profile $(\pi_i^*)_{i\in N}$ which satisfies the following condition for all players $i \in N$ and for all strategies $\pi_i: T \to S_i$:
\begin{equation}\label{eq:def-bne}
    \int_{T^n} u_i\bigl(\pi_i^*(t_i), \pi_{-i}^*(t_{-i})\bigr)\,p(t_i, t_{-i})\,dt
    \;\ge\;
    \int_{T^n} u_i\bigl(\pi_i(t_i), \pi_{-i}^*(t_{-i})\bigr)\,p(t_i, t_{-i})\,dt.
\end{equation}

\begin{thm}\label{thm:bayesian_nash_equilibrium}
    $\forall i\in N,\; \pi_i^*: t_i\mapsto t_i$ is the unique symmetric Bayesian Nash equilibrium.
\end{thm}

From Theorem~\ref{thm:nash_equilibria}, for an arbitrary type profile $t \in T^n$, going as fast as possible is included in rational behaviour.
If a player deviates from it and chooses his ETA as $s_i > t_i$, there is a positive probability that some other player $j$ with $s_i > s_j = t_j > t_i$ exists.
In this case, from Lemma~\ref{lem:faster_better}, player $i$ could have improved his utility by going as fast as possible.
To avoid this kind of risk, going as fast as possible arises as a rational behaviour even in the incomplete information situation.
Regarding uniqueness, we focus on symmetric Bayesian Nash equilibria and show that there is no symmetric Bayesian Nash equilibrium other than $(\pi_i^*)_{i \in N}$.
The detailed explanation appears in \ref{app5}.

Theorem~\ref{thm:bayesian_nash_equilibrium} suggests that in the real world where players cannot perceive other players' types when they make decisions, SFTW is the only rational outcome among symmetric behaviours.
Thus, we have succeeded in game-theoretically explaining the SFTW phenomenon as the unique symmetric Bayesian equilibrium outcome.

\subsection{Illustrative Example}
To visualise the theoretical findings derived in this section, consider a scenario with four vessels ($N=\{A, B, C, D\}$) ordered by their types $t_A < t_B < t_C < t_D$ without any queue.

First, we consider the case of incomplete information.
Consistent with Theorem~\ref{thm:bayesian_nash_equilibrium}, rational players cannot coordinate their arrivals due to the lack of information regarding others' capabilities.
Consequently, every vessel chooses to arrive at its earliest feasible time ($s_i = t_i$).
As illustrated in Figure~\ref{fig:eq_incomplete}, this strategy profile results in immediate queue formation.
For example, Vessel $B$ arrives at $t_B$ while Vessel $A$ is still being serviced, necessitating anchorage waiting until the berth becomes free.
Similarly, Vessel $C$ incurs waiting time.
This outcome reflects the SFTW phenomenon where players rush to secure a position in the FCFS environment.

\begin{figure}[htbp]
    \centering
    \includegraphics[width=\linewidth]{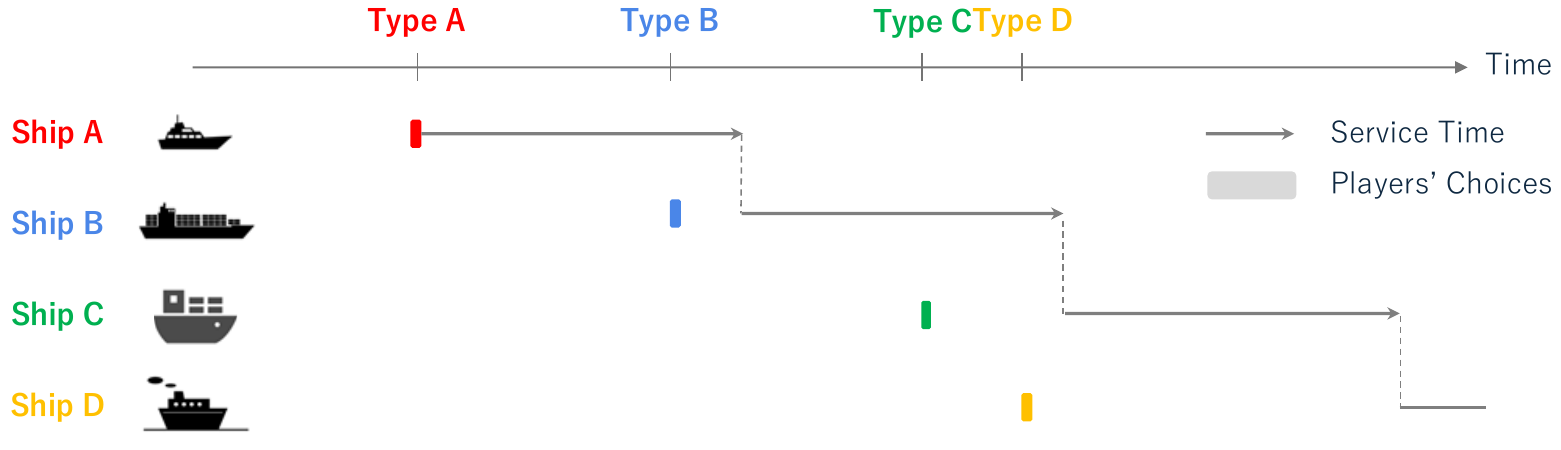}
    \caption{Equilibrium Arrival Pattern under Incomplete Information}
    \label{fig:eq_incomplete}
\end{figure}

Next, we examine the case of complete information.
As derived in Theorem~\ref{thm:nash_equilibria}, the shared knowledge of types expands the set of equilibria for each vessel.
Figure~\ref{fig:eq_complete} visualises these expanded sets: the colored bars represent the entire range of equilibrium arrival times $\Theta_i$ for each vessel.
For instance, Vessel $B$ represents a scenario where the equilibrium strategy is not a single point but a continuous interval $[t_B, \theta_B]$.
Any arrival time within this width maintains the same service order and completion time, making it a valid best response.
Similarly, Vessel $C$'s equilibrium set spans the interval $[t_C, t_D]$.
This indicates that under complete information, players gain significant flexibility: they are no longer forced to rush to $t_i$ but can rationally choose any timing within these expanded intervals without suffering utility loss.

\begin{figure}[htbp]
    \centering
    \includegraphics[width=\linewidth]{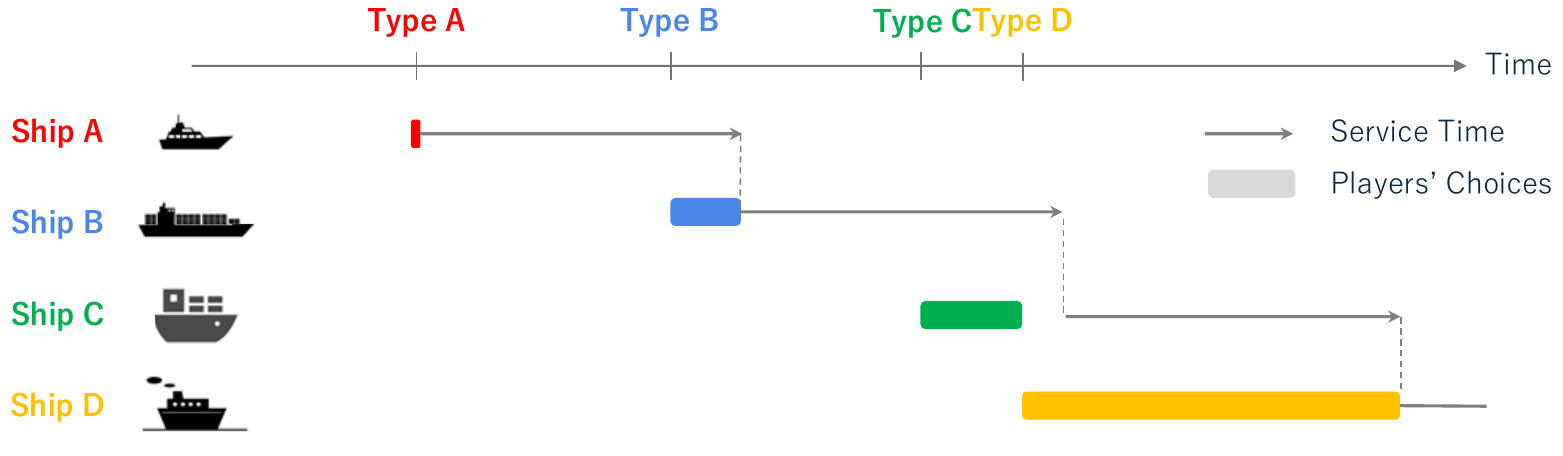}
    \caption{Equilibrium Arrival Pattern under Complete Information}
    \label{fig:eq_complete}
\end{figure}

\section{Evaluation of Information Sharing in the Real World Shipping}\label{sec:empirical}
\subsection{Equilibrium Selection in Complete Information Environments}\label{subsec:eq_selection}
As demonstrated in Theorem~\ref{thm:bayesian_nash_equilibrium}, under incomplete information, SFTW emerges as the unique symmetric Bayesian Nash equilibrium.
Conversely, Theorem~\ref{thm:nash_equilibria} shows that sharing private information $t_1, \dots, t_n$ expands the set of rational behaviours.
Specifically, for any player $i$, there exists a range of arrival times $\Theta_i$ that constitute best responses.
Importantly, within our model, any strategy $s_i \in \Theta_i$ results in the exact same service completion time.
This implies that players are strictly 
indifferent regarding their utility $u_i$ across all strategies in $\Theta_i$.

To predict which equilibrium will be realized from this indifferent set, we invoke Schelling's concept of focal points~\cite{Schelling_1980}.
In situations where multiple equilibria yield identical payoffs, players coordinate their expectations based on ``salience,'' a feature that stands out due to shared cultural or institutional context.
We argue that the strategy corresponding to the latest possible arrival time, $(\sup\Theta_i)_{i\in N}$, serves as this focal point, identified as ``Green Navigation.''

In the contemporary shipping industry, environmental stewardship has transcended abstract ethics to become a dominant institutional norm~\citep{Psaraftis_2019}.
Reinforced by global frameworks such as the IMO's GHG strategy and increasing pressure for supply chain decarbonization, the principle of avoiding emissions provides a distinct and universally intelligible signal for coordination~\cite{Sugden_2006, Suchanek_2018}.
Because the strategy $\sup\Theta_i$ represents the most environmentally friendly action within the feasible set, it is the uniquely salient candidate that actors can reasonably expect others to converge upon.

Therefore, we conclude that the realized strategy profile under information sharing is $(\sup\Theta_i)_{i\in N}$.
This selection is driven not by a change in the payoff structure, but by the coordinating power of established industry norms.
In other words, information sharing enables ships to synchronize their actions with sustainability goals without sacrificing their primary utility.

Having established that the strategy profile $(\sup\Theta_i)_{i\in N}$ serves as the focal equilibrium for Green Navigation, we now quantify the operational flexibility it provides.
We define ``slack'' as the time difference between the arrival time in the SFTW equilibrium $t_i$ and the Green Navigation equilibrium $\sup\Theta_i$.

Figure~\ref{fig:slack_definition} illustrates this concept.
The white bar represents the time duration that a vessel can rationally delay its arrival to sail slower by utilising information about other vessels' constraints.
This interval corresponds to the period during which the berth is occupied by preceding vessels despite vessel $i$ being physically capable of arriving.
Crucially, utilizing this slack for slow steaming does not delay the service start time, allowing for fuel savings without compromising logistical efficiency.

\begin{figure}[htbp]
    \centering
    \includegraphics[width=0.9\linewidth]{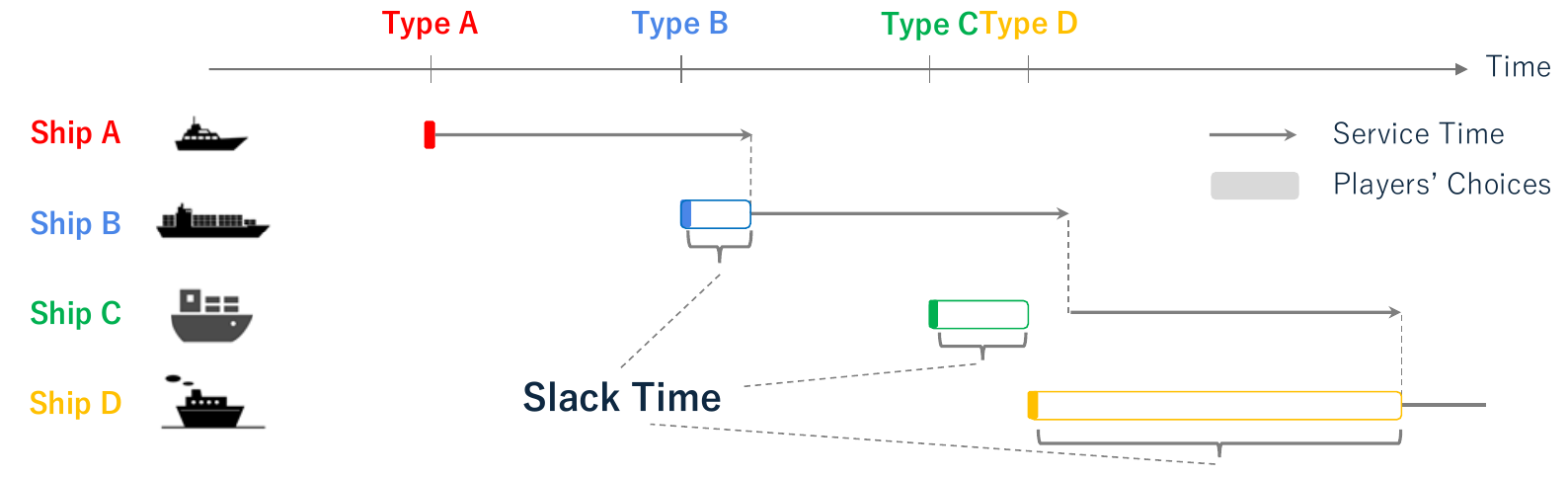}
    \caption{Visualisation of Slack: Potential for Slow Steaming}
    \label{fig:slack_definition}
\end{figure}

\subsection{AIS Data Processing and Parametrization}
To quantitatively evaluate our theoretical findings, we construct an AIS-based dataset of calls to Port Hedland (PHE) in Australia.
Port Hedland is one of the world's largest bulk export ports and is suitable for our case study because it operates under a protocol that necessitates anchoring for queue management, resembling the FCFS competition modeled in our game.

As a sampling frame, we first identify all vessels that exhibit at least one prolonged stop in the vicinity of PHE during the period from 1 January 2022 to 31 December 2024.
Specifically, using raw AIS data sampled at 3-hour intervals, we detect vessels that remain within a coarse circle of radius $0.5^{\circ}$ around the nominal port centre with a speed not exceeding $v_{\text{stop}} = 2$ knots for at least $6$ consecutive hours.
We treat these vessels as having called at PHE at least once in the sample period and use this criterion solely to define the relevant population.
We then restrict attention to ships whose static AIS information classifies them as dry bulk carriers, ensuring that the empirical analysis focuses on iron ore trades operated by vessels of similar size and handling characteristics.
For each vessel, we assemble the AIS time series, sort messages by timestamp, remove invalid records, and truncate implausible speed values outside the range of $0$ to $40$ knots.
The resulting dataset includes $9,539$ voyage records identified as calls to PHE.

For the queuing analysis, we adopt an operational definition of arrival corresponding to the moment a vessel joins the congestion system.
Using the haversine formula with the port centre at latitude $-20.31$ and longitude $118.57$, we compute the distance to PHE for every record.
We define a binary indicator $I_{\text{port}}(t)$ that equals one when the vessel is within $R_{\text{port}} = 60$ km of the port centre. 
Port entry events are identified as the first timestamps at which $I_{\text{port}}(t)$ switches from zero to one.
We denote this time by $t^{\text{entry}}$; in our model, $t^{\text{entry}}$ represents the arrival time at the single-server system.
For each entry, the departure time from the previous port, denoted by $\tau$, is inferred using a hybrid rule.
Looking backwards from $t^{\text{entry}}$, we identify the final stop at the previous port as the last continuous episode outside the PHE geofence where the speed $v(t)$ remains below $v_{\text{stop}}$ knot for at least $6$ hours.
From the end of this stop, we search forward for the first instance where speed exceeds an acceleration threshold $v_{\text{go}} = 8$ knots, setting its timestamp as $\tau$.
In the absence of a clear acceleration event, the end of the stop is used as $\tau$.
The inbound voyage is then defined as the segment of AIS records between $\tau$ and $t^{\text{entry}}$, allowing us to compute the sailing duration and distance.

Within PHE, we identify berthing events to measure service times.
We define a berthing indicator $I_{\text{berth}}(t)$ that equals one when the vessel is within $R_{\text{berth}} = 3$ km of the port centre with speed below $v_{\text{stop}}$ knot.
We detect contiguous intervals where $I_{\text{berth}}(t)=1$ for at least $6$ hours, treating each block as a continuous berth stay.
Applying filters to remove incomplete or unrealistic records (duration outside $3$ to $500$ hours), we identify the earliest valid berthing block for each voyage.
The cargo handling time, $T^{\text{service}}$, is defined as the duration of this block.
The distribution of $T^{\text{service}}$ has a median of $36.00$ hours, and a $10$ percent trimmed mean yields a robust estimate of approximately $35.94$ hours.

A critical challenge in applying our single-server model to a real-world port with multiple berths is the definition of the service time parameter, $\gamma$.
In reality, the port processes multiple ships simultaneously.
However, from the perspective of a ship joining the queue, the essential variable is not the physical cargo handling time of a specific berth, but the rate at which the system clears the queue, i.e., the interval at which a berth becomes available.
Therefore, we approximate the port's capacity as a single aggregated server with an effective service time $\gamma$.
We compute the inter-arrival times between consecutive berth start events across the entire port, restricting the analysis to positive intervals larger than 0.1 hours.
The 10 percent trimmed mean of these intervals measures the typical time between successive vessel acceptances by the port.
In our data, this value is approximately $3.97$ hours.
We therefore adopt $\gamma \approx 3.97$ hours as the calibrated service time parameter for the single-server approximation in our theoretical analysis.
This approximation allows us to capture the macroscopic congestion dynamics and the strategic interactions among vessels without getting obscured by the complexities of specific berth allocations.

The combination of an average cargo handling time of about $36$ hours and a throughput based service time of about $3.97$ hours implies that, in steady state, the port operates as if roughly nine berths were continuously occupied by the dry bulk carriers in our sample.
This effective number of parallel servers is smaller than the total number of physical deep water iron ore berths in the inner harbour, which is on the order of nineteen, but this is precisely what one would expect in practice.
Not all iron ore berths are continuously assigned to large capesize vessels of the type we observe, some berths are intermittently used by smaller bulkers or other commodities, and some capacity is periodically unavailable due to tides, maintenance, or operational constraints.
The fact that the effective utilisation corresponds to somewhat more than half of the iron ore capable berths is therefore consistent with the actual multi berth operation at PHE rather than indicative of any inconsistency between the observed handling times and the inferred port level throughput.

% We analysed AIS data covering the period from January 1, 2022, to December 31, 2024.
% To ensure the validity of the single-server approximation with a constant service rate $\gamma$, we restricted our sample exclusively to dry bulk carriers arriving in ballast condition to load iron ore. Since these vessels share similar dimensions (mostly Capesize) and operational characteristics, the variance in their handling times is minimized, justifying the use of a representative aggregate service time.

% The dataset includes $9,539$ voyage records identified as calls to PHE.
% We defined the ``arrival at the port'' as the moment a vessel crosses a geofence with a radius of 60 km from the port centre.
% This geofence was determined by AIS data as speed over ground (SOG) drastically dropped down to around zero and vessels began anchoring in this vicinity.

% This approximation allows us to capture the macroscopic congestion dynamics and the strategic interactions among vessels without getting obscured by in the complexities of specific berth allocations.

\subsection{Observation of Excessive Queuing}
Before applying our counterfactual analysis, we examine the current state of congestion at PHE to see if it aligns with the characteristics of the SFTW equilibrium.
Theorem~\ref{thm:bayesian_nash_equilibrium} predicts that under incomplete information and FCFS competition environments, vessels maximize their utility by securing the earliest possible position in the queue, regardless of the queue length.
Consequently, we expect to observe substantial waiting times that cannot be explained solely by service time variability.

We calculated the comprehensive offshore waiting time for each vessel, defined as the duration from the arrival at the port area (60km radius) to the commencement of berthing. 
The data reveals that vessels experience a mean waiting time of $100.73$ hours and a median of $75.00$ hours.
As shown in Figure~\ref{fig:offshore_waiting}, the distribution of waiting times is heavily skewed with a long tail.

\begin{figure}
    \centering
    \includegraphics[width=\linewidth]{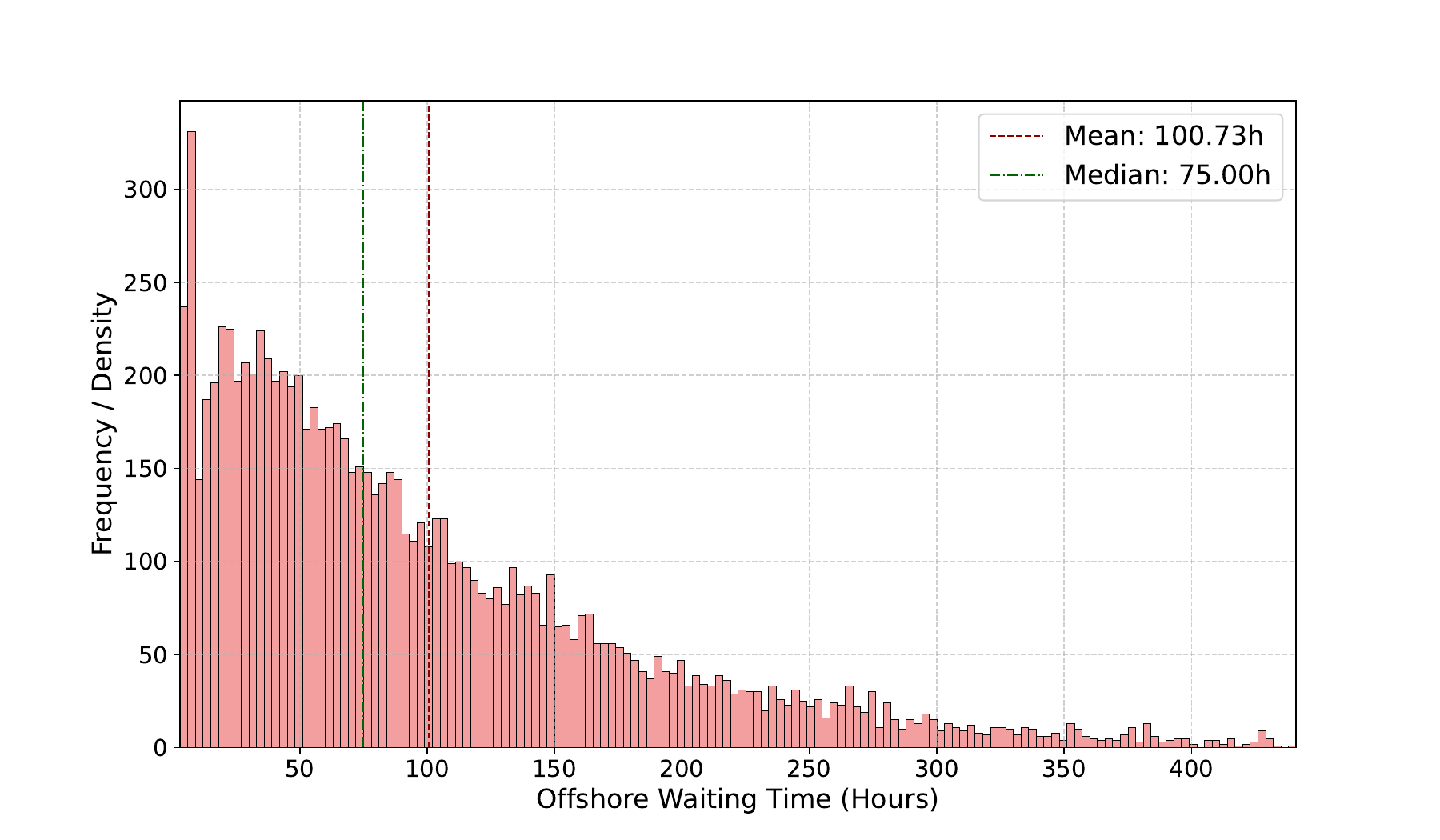}
    \caption{Distribution of Comprehensive Offshore Waiting Times. The heavy-tailed distribution with a mean of $100.73$ hours indicates excessive queuing consistent with SFTW behaviour.}
    \label{fig:offshore_waiting}
\end{figure}

Comparing the mean waiting time ($100.73$ hours) with the effective service time ($\gamma \approx 3.97$hours), the vessels spend, on average, waiting roughly $25$ times longer than the system's inter-service interval.
If vessels were coordinating their arrivals or operating under a mechanism that allowed them to delay arrival without penalty (as in the complete information equilibrium), such excessive queuing would be rationalised into slow steaming.
Despite the clear inefficiency, the persistence of these massive waiting times suggests that vessels are not adjusting their arrival times to match the port's availability.
Instead, they are engaging in a non-cooperative competition to arrive as early as possible to secure a spot in the FCFS queue.
This observation is consistent with the SFTW equilibrium predicted by our theoretical model.

\subsection{Quantifying the Value of Information}
Finally, we simulate the potential reduction in congestion if information were shared among vessels.
Based on Theorem~\ref{thm:bayesian_nash_equilibrium}, we interpret the observed actual arrival times as the vessels' types (physical earliest feasible arrival times, $t_i$), assuming that ships are currently playing the unique symmetric Bayesian Nash equilibrium where $s_i = t_i$.

We then calculate the theoretical ``slack'' $\delta_i$ for each vessel $i$ under the complete information equilibrium derived in Theorem~\ref{thm:nash_equilibria} and the discussion in Section~\ref{subsec:eq_selection}.
The slack represents the amount of time a ship could have delayed its arrival (i.e., sailed slower) without altering the service order or delaying its service completion time.
The slack is calculated as $\delta_i = \min(t_{i+1}, \theta_i) - t_i$, where $\theta_i$ is the service completion time of the preceding vessel, determined recursively as $\theta_i = \max(t_i, \theta_{i-1}) + \gamma$, with $\gamma = 3.97$ hours.

The simulation results indicate heterogeneous scope for ``Green Navigation'' through information sharing.
Specifically, the analysis yields a mean slack of $1.54$ hours and a median slack of $0.00$ hours per voyage.
Over the three-year period analysed, the total potential slack amounts to $14,721$ hours. 
Figure~\ref{fig:slack_hours} illustrates that most voyages have little or no slack, while a subset exhibits positive slack.

\begin{figure}
    \centering
    \includegraphics[width=\linewidth]{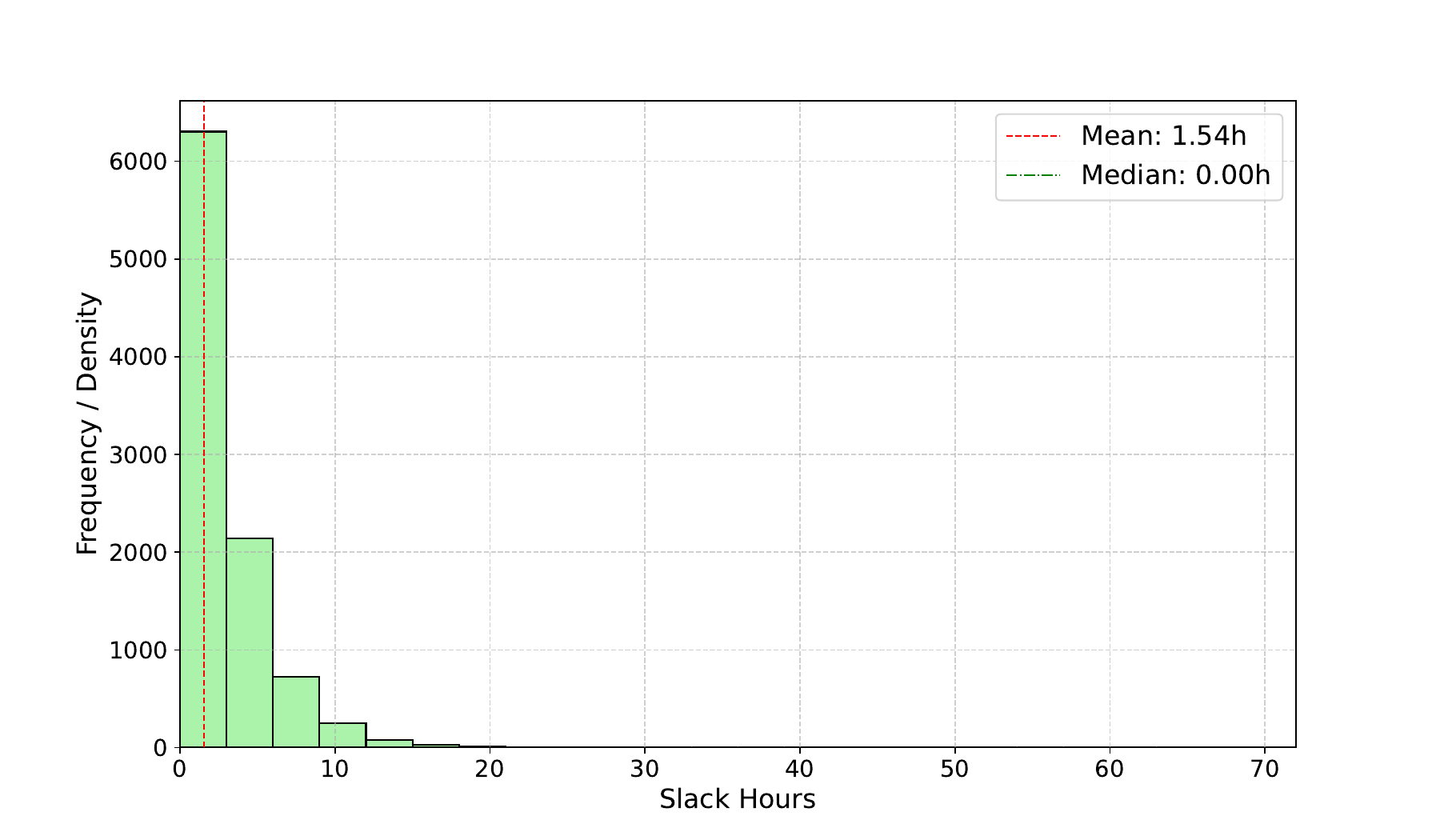}
    \caption{Distribution of Calculated Slack Hours (Potential Slow Steaming Duration)}
    \label{fig:slack_hours}
\end{figure}

The key implication of slack is structural rather than merely quantitative.
By construction, slack measures the time window over which a vessel can delay arrival while preserving the FCFS service order and, in the complete-information equilibrium, preserving its service completion time.
In this sense, slack operationalizes the coordination margin that remains compatible with the existing FCFS regime.
The empirical distribution in Figure~\ref{fig:slack_hours} shows that this margin is highly heterogeneous: it is absent for many voyages, but it emerges for a non-trivial subset.
This heterogeneity is informative in itself, because it indicates when and where information transparency can unlock utility-preserving speed reductions, and it provides an empirical benchmark for the prevalence of such opportunities in a highly congested bulk port.

Moreover, even moderate speed reductions can translate into meaningful fuel and emissions reductions due to the convex relationship between speed and propulsion fuel consumption, often approximated as cubic in speed~\cite{Psaraftis_2023}.
Accordingly, identifying the feasible arrival-delay window is a first step toward quantifying the decarbonisation potential of coordination under FCFS.
Our results therefore connect a strategic mechanism, information-induced equilibrium shifts, to an empirically measurable margin that can be mapped to operational adjustments.

% While an average slack of $1.54$ hours might appear modest compared to the 100-hour waiting time, it represents a pure efficiency gain.
% This time can be directly converted into slow steaming without any negative impact on logistics schedules.
% Considering the cubic relationship between speed and fuel consumption, even a few hours of deceleration from maximum speed can yield significant GHG reductions.
% Crucially, this result demonstrates that even in a highly congested port where the arrival rate closely matches the service rate, simply reducing information incompleteness allows for a utility-preserving adjustment of arrival times, mitigating the SFTW phenomenon.

\section{Discussion}\label{sec:discussion}
Our empirical analysis revealed a striking contrast: while vessels spend an average of over $100$ hours waiting, the theoretical ``slack'' available for slow steaming is only about $1.54$ hours, with a median of $0.00$ hours.

This counter-intuitive result that immense queues do not translate into immense flexibility can be explained by the high density of vessels' types (earliest feasible arrival times).
As derived in Theorem~\ref{thm:nash_equilibria}, the feasible slack for slow steaming is bounded not directly by the queue length, but by the interval between a vessel's own capability $t_i$ and that of the previous and next competitors $t_{i-1}, t_{i+1}$.
In the analyzed data, these arrival capabilities are extremely dense.
In such a dense environment, shifting an arrival time even slightly later could result in losing one's position in the service order, which vessels strictly avoid under the faster the better incentives.
This result can be explained by the extremely high traffic density at one of the most congested ports, PHE.

Consequently, our findings from PHE should be interpreted as a conservative benchmark for a highly congested bulk port, rather than a universal estimate of the coordination gains under FCFS.
In other ports characterized by moderate congestion, the gaps between service slots and between types may be larger.
In such environments, the value of information sharing could be more pronounced, as vessels may have wider windows to adjust their speed without compromising their service order.

Additionally, the small magnitude of the slack does not diminish its importance.
In maritime engineering, fuel consumption is commonly modelled as a convex and highly speed-sensitive function, traditionally close to the cube of speed~\cite{Psaraftis_2023}.
Consequently, speed reductions yield the highest marginal fuel savings when applied to the highest speed ranges.
The derived slack allows vessels to shave off the peak speeds during their voyage.
Even a reduction of a few hours in navigation time, when converting a full-speed sprint into a slightly more moderate pace, can result in meaningful reductions in fuel consumption and greenhouse gas emissions.
Crucially, this environmental benefit is achieved without any loss of utility for the ship operators, as the service completion time remains unchanged.

The findings suggest a clear pathway for policy intervention.
Currently, the industry often looks towards sophisticated ``Virtual Arrival'' contracts or slot auctioning systems to solve congestion~\cite{Senss_2023}.
While effective, these require complex institutional changes, such as contractual renegotiation, cross-stakeholder governance, and strong incentive alignment, which limits their practical adoption.
Our model demonstrates that simply enriching the information environment, specifically, making vessels' types visible to one another, can trigger a shift to a more sustainable equilibrium within the existing FCFS framework.
Technologies that share real-time capability data, such as the implementation of port call optimisation platforms or enhanced AIS protocols can serve as soft infrastructure that enables this coordination.

Several limitations should be noted.
First, we approximated the multi-berth port as a single-server system.
While this captures the macroscopic bottleneck, it simplifies specific berth-allocation dynamics.
Relatedly, our empirical construction uses port-level AIS sequences, which may pool heterogeneous operational patterns across terminals within the port, potentially affecting the measured density of types and hence the implied slack distribution.
Second, our calculation of slack is deterministic and does not account for weather uncertainties or mechanical failures that might affect voyage times.
Finally, while we quantified the time slack, we did not explicitly estimate the GHG reduction volume, as this depends heavily on individual vessel engine characteristics and weather conditions, which were outside the scope of our utility model.

It is worth noting that our theoretical model deliberately adopts a parsimonious utility benchmark in which each vessel's payoff depends only on service completion time.
This simplification is not meant to claim that sailing time and waiting time are economically equivalent in practice; rather, it isolates the strategic mechanism through which information transparency can relax the incentive to race under FCFS without changing the realized service order.
Incorporating additional real-world frictions, such as contractual terms, schedule-reliability concerns, or operating costs incurred while waiting, is an important direction for future work and may further refine how the measured slack translates into implementable speed reductions and emissions outcomes.

\section{Conclusion}\label{sec:conclusion}
This study addressed the persistent ``Sail Fast, Then Wait'' (SFTW) phenomenon in maritime logistics through a game-theoretic lens.
We theoretically demonstrated that SFTW emerges as the unique symmetric Bayesian Nash equilibrium under incomplete information in FCFS environments.
Conversely, we proved that sharing information regarding vessels' feasible arrival times expands the set of equilibria, allowing for a ``Green Navigation'' strategy where vessels voluntarily slow down without delaying their cargo service.

Empirically, using AIS data for vessels calling at PHE, we observed excessive waiting times (averaging over 100 hours) relative to the service interval (around 4 hours), a finding consistent with the non-cooperative SFTW equilibrium.
Our counterfactual simulation revealed that information sharing could uncover an average of $1.54$ hours of slack per voyage.
While seemingly small, this slack represents a ``free lunch''; an opportunity to reduce emissions through slow steaming without compromising the private incentives (service completion time) of the operators.
Furthermore, given that this study focused on a highly saturated port, future research extending this analysis to ports with moderate congestion levels is expected to reveal even larger potential for emissions reduction.

In conclusion, our work provides a theoretical foundation for the value of information in sustainable shipping.
It suggests that before (or alongside) reforming contractual clauses, the industry should prioritize the transparency of operational data.
By removing the uncertainty regarding competitors' capabilities, we can align individual rationality with collective sustainability goals.

% \bibliographystyle{unsrtnat}
% \bibliography{references}  %%% Uncomment this line and comment out the ``thebibliography'' section below to use the external .bib file (using bibtex) .


\begin{thebibliography}{00}

\bibitem{MacGregor_1983} MacGregor, D. R. (1983). The tea clippers: Their history and development 1833–1875. Conway Maritime Press.

\bibitem{Sung_2022} Sung, I., Zografakis, H., \& Nielsen, P. (2022). Multi-lateral ocean voyage optimization for cargo vessels as a decarbonization method. Transportation Research Part D: Transport and Environment, 110, 103407. https://doi.org/10.1016/j.trd.2022.103407

\bibitem{Alvarez_2010} Alvarez, J. F., Longva, T., \& Engebrethsen, E. S. (2010). A methodology to assess vessel berthing and speed optimization policies. Maritime economics \& logistics, 12(4), 327-346. https://doi.org/10.1057/mel.2010.11

\bibitem{Adland_2018} Adland, R. O., \& Jia, H. (2018). Contractual barriers and energy efficiency in the crude oil supply chain. In 2018 IEEE International Conference on Industrial Engineering and Engineering Management, 1-5. https://doi.org/10.1109/IEEM.2018.8607536

\bibitem{IMO_2023_GHG} International Maritime Organization. (2023). 2023 IMO strategy on reduction of GHG emissions from ships [Resolution MEPC.377(80)]. Retrieved from https://wwwcdn.imo.org/localresources/en/KnowledgeCentre/IndexofIMOResolutions/MEPCDocuments/MEPC.377(80).pdf. Accessed September 18, 2025

\bibitem{BIMCO_JIT_2021} BIMCO. (2021). Just in time arrival clause for voyage charter parties 2021. Retrieved from https://www.bimco.org/contractual-affairs/bimco-clauses/current-clauses/just-in-time-arrival-clause-for-voyage-charter-parties-2021/. Accessed September 18, 2025

\bibitem{BlueVisby_MEPC82} Blue Visby Consortium. (2024). GHG emissions reductions through mitigating the operational and commercial practice of Sail Fast Then Wait [Information paper MEPC 82/INF.32]. Retrieved from https://maritimecyprus.com/wp-content/uploads/2024/10/Blue-Visby-Solution\_c.pdf. Accessed September 18, 2025


\bibitem{Jimenez_2021} Jim\'{e}nez, P., Cano, J. C., Calafate, C. T., \& Manzoni, P. (2021). HADES: A multi-agent platform to reduce congestion anchoring at the Port of Cartagena. \textit{Applied Sciences}, 11(7), 3109. https://doi.org/10.3390/app11073109

\bibitem{Wang_2025} Wang, W., Wang, H., Pang, K. W., Zhen, L., \& Wang, S. (2025). Optimizing bunkering and sailing strategies for sustainable shipping: a decision model for reducing costs and carbon emissions. Annals of Operations Research, 1-19. https://doi.org/10.1007/s10479-025-06650-4

\bibitem{Zhen_2025} Zhen, L., Jiang, M., Wang, S., \& Wu, J. (2025). Ship fleet scheduling and deployment optimization considering carbon and sulfur emissions. European Journal of Operational Research. (In Press) https://doi.org/10.1016/j.ejor.2025.07.021

\bibitem{Jensen_2025} Omholt-Jensen, S., Fagerholt, K., \& Meisel, F. (2025). Fleet repositioning in the tramp ship routing and scheduling problem with bunker optimization: A matheuristic solution approach. European Journal of Operational Research, 321(1), 88-106. https://doi.org/10.1016/j.ejor.2024.09.029

% \bibitem{Bu_2023} Bu, F., Liu, J., Liao, H., \& Nachtmann, H. (2023). An alternative solution to congestion relief of U.S. seaports by container-on-barge: A simulation study. \textit{Simulation Modelling Practice and Theory}, 129, 102836. https://doi.org/10.1016/j.simpat.2023.102836


% \bibitem{ref4} Kadıo\v{g}lu, S. O., et al. (2022). Real-time ETA sharing and AI-assisted berth planning for port efficiency. \textit{Maritime Transport Research}, 3, 100052. https://doi.org/10.1016/j.martra.2022.100052

% \bibitem{vandersteeg_2023} van der Steeg, J. J., Benders, R., \& van der Meer, R. (2023). Berth planning and real-time disruption recovery: A simulation study for a tidal port. \textit{Flexible Services and Manufacturing Journal}, 35, 70--110. https://doi.org/10.1007/s10696-021-09420-w

% \bibitem{Jiang_2025} Jiang, S., Liu, L., Peng, P., Xu, M., \& Yan, R. (2025). Prediction of vessel arrival time to port: a review of current studies. Maritime Policy \& Management, 1–26. https://doi.org/10.1080/03088839.2025.2488376

% \bibitem{Sheikholeslami_2013} Sheikholeslami, A., Tabari, A., et al. (2013). Practical solutions for reducing container ships' waiting times at ports using simulation model. \textit{Journal of Marine Science and Application}, 12(4), 434--444. https://doi.org/10.1007/s11804-013-1214-x

% \bibitem{Bakker_2024} Bakker, F. P., et al. (2024). Port accessibility depends on cascading interactions between fleets, policies, infrastructure, and hydrodynamics. \textit{Journal of Marine Science and Engineering}, 12(6), 1006. https://doi.org/10.3390/jmse12061006


\bibitem{Hassin_2003} Hassin, R., \& Haviv, M. (2003). To queue or not to queue: Equilibrium behaviour in queueing systems (Vol. 59). Springer Science \& Business Media.

\bibitem{Haviv_2021} Haviv, M., \& Ravner, L. (2021). A survey of queueing systems with strategic timing of arrivals. Queueing Systems, 99(1), 163-198. https://doi.org/10.1007/s11134-021-09717-8

\bibitem{Vickrey_1969} Vickrey, W. S. (1969). Congestion theory and transport investment. American Economic Review, 59(2), 251-260. https://www.jstor.org/stable/1823678

\bibitem{Arnott_1990} Arnott, R., De Palma, A., \& Lindsey, R. (1990). Economics of a bottleneck. Journal of Urban Economics, 27(1), 111-130. https://doi.org/10.1016/0094-1190(90)90028-L





\bibitem{Arnott_1993} Arnott, R., De Palma, A., \& Lindsey, R. (1993). A structural model of peak-period congestion: A traffic bottleneck with elastic demand. American Economic Review, 161-179. https://www.jstor.org/stable/2117502

\bibitem{Palma_2013} De Palma, A., \& Fosgerau, M. (2013). Random queues and risk averse users. European Journal of Operational Research, 230(2), 313-320. https://doi.org/10.1016/j.ejor.2013.04.041


\bibitem{Platz_2017} Platz, T. T., \& \O sterdal, L. P. (2017). The curse of the first-in–first-out queue discipline. Games and Economic behaviour, 104, 165-176. https://doi.org/10.1016/j.geb.2017.03.004

\bibitem{Wu_2019} Wu, C., Bassamboo, A., \& Perry, O. (2019). Service system with dependent service and patience times. Management Science, 65(3), 1151-1172. https://doi.org/10.1287/mnsc.2017.2983

\bibitem{Li_2020} Li, Z. C., Huang, H. J., \& Yang, H. (2020). Fifty years of the bottleneck model: A bibliometric review and future research directions. Transportation research part B: methodological, 139, 311-342. https://doi.org/10.1016/j.trb.2020.06.009

\bibitem{Yu_2025} Yu, X., van den Berg, V. A., \& Verhoef, E. T. (2025). Preference heterogeneity in a dynamic flow congestion model. Transportation Research Part B: Methodological, 195, 103193. https://doi.org/10.1016/j.trb.2025.103193


\bibitem{Rapoport_2004} Rapoport, A., Stein, W. E., Parco, J. E., \& Seale, D. A. (2004). Equilibrium play in single-server queues with endogenously determined arrival times. Journal of Economic behaviour \& Organization, 55(1), 67-91. https://doi.org/10.1016/j.jebo.2003.07.003

\bibitem{Otsubo_2008} Otsubo, H., \& Rapoport, A. (2008). Vickrey's model of traffic congestion discretized. Transportation Research Part B: Methodological, 42(10), 873-889. https://doi.org/10.1016/j.trb.2008.04.002

\bibitem{Rapoport_2010} Rapoport, A., Stein, W. E., Mak, V., Zwick, R., \& Seale, D. A. (2010). Endogenous arrivals in batch queues with constant or variable capacity. Transportation Research Part B: Methodological, 44(10), 1166-1185. https://doi.org/10.1016/j.trb.2010.01.005

\bibitem{Breinbjerg_2016} Breinbjerg, J., Sebald, A., \& \O sterdal, L. P. (2016). Strategic behaviour and social outcomes in a bottleneck queue: experimental evidence. Review of Economic Design, 20(3), 207-236. https://doi.org/10.1007/s10058-016-0190-4

\bibitem{Breinbjerg_2017} Breinbjerg, J. (2017). Equilibrium arrival times to queues with general service times and non-linear utility functions. European Journal of Operational Research, 261(2), 595-605. https://doi.org/10.1016/j.ejor.2017.03.010

\bibitem{Sakuma_2020} Sakuma, Y., Masuyama, H., \& Fukuda, E. (2020). A discrete-time single-server Poisson queueing game: Equilibria simulated by an agent-based model. European Journal of Operational Research, 283(1), 253-264. https://doi.org/10.1016/j.ejor.2019.11.003

\bibitem{Alon_2022} Alon, T., \& Haviv, M. (2022). Discrete-time strategic job arrivals to a single machine with waiting and lateness penalties. European Journal of Operational Research, 303(1), 480-486. https://doi.org/10.1016/j.ejor.2022.02.032

\bibitem{Alon_2023} Alon, T., \& Haviv, M. (2023). Choosing a batch to be processed. Annals of Operations Research, 326(1), 67-87. https://doi.org/10.1007/s10479-023-05274-w

\bibitem{Breinbjerg_2024} Breinbjerg, J., Platz, T. T., \& \O sterdal, L. P. (2024). Equilibrium arrivals to a last-come first-served preemptive-resume queue. Annals of Operations Research, 336(3), 1551-1572. https://doi.org/10.1007/s10479-023-05348-9

\bibitem{Ha_2001} Ha, A. Y. (2001). Optimal pricing that coordinates queues with customer-chosen service requirements. Management Science, 47(7), 915-930. https://doi.org/10.1287/mnsc.47.7.915.9806

\bibitem{Afeche_2013} Afeche, P. (2013). Incentive-compatible revenue management in queueing systems: Optimal strategic delay. Manufacturing \& Service Operations Management, 15(3), 423-443. https://doi.org/10.1287/msom.2013.0449

\bibitem{Gavirneni_2016} Gavirneni, S., \& Kulkarni, V. G. (2016). Self‐selecting priority queues with burr distributed waiting costs. Production and Operations Management, 25(6), 979-992. https://doi.org/10.1111/poms.12520

\bibitem{Wan_2017} Wan, G., \& Wang, Q. (2017). Two‐tier healthcare service systems and cost of waiting for patients. Applied Stochastic Models in Business and Industry, 33(2), 167-183. https://doi.org/10.1002/asmb.2231

\bibitem{Silva_2017} Silva, H. E., Lindsey, R., De Palma, A., \& Van den Berg, V. A. (2017). On the existence and uniqueness of equilibrium in the bottleneck model with atomic users. Transportation Science, 51(3), 863-881. https://doi.org/10.1287/trsc.2016.0672

\bibitem{Maglaras_2018} Maglaras, C., Yao, J., \& Zeevi, A. (2018). Optimal price and delay differentiation in large-scale queueing systems. Management Science, 64(5), 2427-2444. https://doi.org/10.1287/mnsc.2016.2713

\bibitem{Islier_2024} \.{I}\c{s}lier, Z. G., \& G\"{u}ll\"{u}, R. (2024). On strategic customers with correlated utility attributes: Effects and information benefits. European Journal of Operational Research, 313(1), 258-269. https://doi.org/10.1016/j.ejor.2023.08.015

\bibitem{Tuncalp_2024} Tun\c{c}alp, F., G\"{u}ne\c{s}, E. D., \& \"{O}rmeci, E. L. (2024). Modeling strategic walk-in patients in appointment systems: Equilibrium behaviour and capacity allocation. European Journal of Operational Research, 313(2), 587-601. https://doi.org/10.1016/j.ejor.2023.09.006

\bibitem{Naor_1969} Naor, P. (1969). The regulation of queue size by levying tolls. Econometrica, 37(1), 15-24. https://doi.org/10.2307/1909200

\bibitem{Edelson_1975} Edelson, N. M., \& Hilderbrand, D. K. (1975). Congestion tolls for Poisson queuing processes. Econometrica, 43(1), 81-92. https://doi.org/10.2307/1913415

\bibitem{Economou_2022} Economou, A. (2022). How much information should be given to the strategic customers of a queueing system?. Queueing Systems, 100(3), 421-423. https://doi.org/10.1007/s11134-022-09741-2


\bibitem{Hassin_1986} Hassin, R. (1986). Consumer information in markets with random product quality: The case of queues and balking. Econometrica, 1185-1195. https://doi.org/10.2307/1912327

\bibitem{Seale_2005} Seale, D. A., Parco, J. E., Stein, W. E., \& Rapoport, A. (2005). Joining a queue or staying out: Effects of information structure and service time on arrival and staying out decisions. Experimental Economics, 8(2), 117-144. https://doi.org/10.1007/s10683-005-0872-1

\bibitem{Stein_2007} Stein, W. E., Rapoport, A., Seale, D. A., Zhang, H., \& Zwick, R. (2007). Batch queues with choice of arrivals: Equilibrium analysis and experimental study. Games and Economic behaviour, 59(2), 345-363. https://doi.org/10.1016/j.geb.2006.08.008

\bibitem{Guo_2007} Guo, P., \& Zipkin, P. (2007). Analysis and comparison of queues with different levels of delay information. Management Science, 53(6), 962-970.https://doi.org/10.1287/mnsc.1060.0686

\bibitem{Simhon_2016} Simhon, E., Hayel, Y., Starobinski, D., \& Zhu, Q. (2016). Optimal information disclosure policies in strategic queueing games. Operations Research Letters, 44(1), 109-113. https://doi.org/10.1016/j.orl.2015.12.005

\bibitem{Kim_2017} Kim, B., \& Kim, J. (2017). Optimal information disclosure policies in a strategic queueing model. Operations Research Letters, 45(2), 181-186. https://doi.org/10.1016/j.orl.2017.02.003

\bibitem{Hu_2018} Hu, M., Li, Y., \& Wang, J. (2018). Efficient ignorance: Information heterogeneity in a queue. Management Science, 64(6), 2650-2671. https://doi.org/10.1287/mnsc.2017.2747

\bibitem{Hassin_2020} Hassin, R., \& Roet-Green, R. (2020). On queue-length information when customers travel to a queue. Manufacturing \& Service Operations Management, 23(4), 989-1004. https://doi.org/10.1287/msom.2020.0909

\bibitem{Dimitrakopoulos_2021} Dimitrakopoulos, Y., Economou, A., \& Leonardos, S. (2021). Strategic customer behaviour in a queueing system with alternating information structure. European Journal of Operational Research, 291(3), 1024-1040. https://doi.org/10.1016/j.ejor.2020.10.054

\bibitem{Economou_2024} Economou, A. (2024). The impact of information about last customer's decision on the join-or-balk dilemma in a queueing system. Annals of Operations Research, 1-24. https://doi.org/10.1007/s10479-024-06262-4

\bibitem{Zhen_2015} Zhen, L. (2015). Tactical berth allocation under uncertainty. European Journal of Operational Research, 247(3), 928-944. https://doi.org/10.1016/j.ejor.2015.05.079

\bibitem{Schelling_1980} Schelling, T., 1980. The Strategy of Conflict. Harvard University Press, Cambridge, MA (first edition 1960).

\bibitem{Psaraftis_2019} Psaraftis, H. N. (2019). Decarbonization of maritime transport: to be or not to be?. Maritime Economics \& Logistics, 21(3), 353-371. https://doi.org/10.1057/s41278-018-0098-8

\bibitem{Sugden_2006} Sugden, R., \& Zamarr\'{o}n, I. E. (2006). Finding the key: the riddle of focal points. Journal of Economic Psychology, 27(5), 609-621. https://doi.org/10.1016/j.joep.2006.04.003

\bibitem{Suchanek_2018} Suchanek, A., \& Entschew, E. M. (2018). Ethical focal points as a complement to accelerated social change. Humanistic Management Journal, 3(2), 221-232. https://doi.org/10.1007/s41463-018-0045-y

\bibitem{Psaraftis_2023} Psaraftis, H. N., \& Lagouvardou, S. (2023). Ship speed vs power or fuel consumption: Are laws of physics still valid? Regression analysis pitfalls and misguided policy implications. Cleaner Logistics and Supply Chain, 7, 100111. https://doi.org/10.1016/j.clscn.2023.100111

\bibitem{Senss_2023} Senss, A., Canbulat, O., Uzun, D., Gunbeyaz, S. A., \& Turan, O. (2023). Just in time vessel arrival system for dry bulk carriers. Journal of Shipping and Trade, 8, 12. https://doi.org/10.1186/s41072-023-00141-0

% \bibitem{Ravner_2024} Ravner, L., \& Snitkovsky, R. I. (2024). Stochastic approximation of symmetric nash equilibria in queueing games. Operations Research, 72(6), 2698-2725. https://doi.org/10.1287/opre.2021.0306

\end{thebibliography}

\appendix
\section{Proof of the Equation~(1)}
\label{app1}

\begin{proof}
    We prove by mathematical induction that the waiting time of any player $i\in N$ is given by Equation \eqref{eq:waiting_time}:
    \begin{equation*}
        w_i^g(s_i, s_{-i}) = \max_{1\le k\le g^{-1}(i)} \{\, s_{g(k)} + (g^{-1}(i)-k)\gamma - s_i \,\}.
    \end{equation*}

    First consider $g^{-1}(i)=1$. No one is ahead of $i$, hence $w_i^g(s_i, s_{-i})=0$. On the other hand,
    \begin{align*}
        \max_{1\le k\le g^{-1}(i)} \{ s_{g(k)} + (g^{-1}(i)-k)\gamma - s_i \}
        &= \max_{1\le k\le 1} \{ s_{g(k)} + (1-k)\gamma - s_i \}\\
        &= s_{g(1)} + (1-1)\gamma - s_i \;=\; 0.
    \end{align*}
    Thus, Equation~\eqref{eq:waiting_time} holds true.
    
    Assume Equation~\eqref{eq:waiting_time} holds for $g^{-1}(i)=j\in\mathbb{N}$. We show it also holds for $g^{-1}(i)=j+1$.
    When $g^{-1}(i)=j+1\ge 2$, service to vessel $i$ starts immediately if there is no queue upon arrival, otherwise it starts when the service of the immediately preceding vessel $g(g^{-1}(i)-1)\in N$ finishes.
    \begin{equation*}
        \therefore w_i^g(s_i, s_{-i})
        = \max\{\, 0,\; s_{g(g^{-1}(i)-1)} + w_{g(g^{-1}(i)-1)}^g(s_i, s_{-i}) + \gamma - s_i \,\}.
    \end{equation*}
    Because $g^{-1}(i)=j+1$ from assumption,
    \begin{equation*}
        w_i^g(s_i, s_{-i})
        = \max\{\, 0,\; s_{g(j)} + w_{g(j)}^g(s_i, s_{-i}) + \gamma - s_i \,\}.
    \end{equation*}
    By the induction hypothesis,
    \begin{align*}
        w_{g(j)}^g(s_i, s_{-i})
        &= \max_{1\le k\le g^{-1}(g(j))} \{\, s_{g(k)} + (g^{-1}(g(j)) - k)\gamma - s_{g(j)} \,\}\\
        &= \max_{1\le k\le j} \{\, s_{g(k)} + (j - k)\gamma - s_{g(j)} \,\}.
    \end{align*}
    Therefore,
    \begin{align*}
        w_i^g(s_i, s_{-i})
        &= \max\Bigl\{\, 0,\; s_{g(j)} + \max_{1\le k\le j} \{\, s_{g(k)} + (j-k)\gamma - s_{g(j)} \,\} + \gamma - s_i \Bigr\}\\
        &= \max\Bigl\{\, 0,\; \max_{1\le k\le j} \{\, s_{g(k)} + (j+1-k)\gamma - s_i \,\} \Bigr\}.
    \end{align*}
    Herein, note that $s_{g(j+1)} + (j+1-(j+1))\gamma - s_i = 0$ because $g(j+1)=i$.
    \begin{align*}
        \therefore w_i^g(s_i, s_{-i}) = &\max_{1\le k\le j+1} \{\, s_{g(k)} + (j+1-k)\gamma - s_i \,\}\\
        &= \max_{1\le k\le g^{-1}(i)} \{\, s_{g(k)} + (g^{-1}(i)-k)\gamma - s_i \,\}.
    \end{align*}
    Thus Equation~\eqref{eq:waiting_time} holds for $g^{-1}(i)=j+1$.
    The statement follows by induction.
\end{proof}

\section{Proof of Lemma~\ref{lem:faster_better}}
\label{app2}
\begin{proof}
    Suppose an arbitrary strategy profile $(s_1,\dots,s_n)$ and let $G(s_i, s_{-i})$ be the set of feasible service orders.
    If player $i\in N$ advances her arrival to some $s_i'<s_i$, let $G(s_i', s_{-i})$ be the resulting set of orders.
    Assume this change advances her expected service position $\mathbb{E}\bigl(g_{(s_i, s_{-i})}^{-1}(i)\bigr) > \mathbb{E}\bigl(g_{(s_i', s_{-i})}^{-1}(i)\bigr)$.
    From this assumption, following two statements holds:
    \begin{equation}\label{eq:minmax_g}
        \forall g'\in G(s_i', s_{-i}),\,\forall g\in G(s_i, s_{-i}),\, g^{-1}(i) \ge g'^{-1}(i),
    \end{equation}
    \begin{equation}\label{eq:noneq_g_g'}
        \exists g'\in G(s_i', s_{-i}),\,\exists g\in G(s_i, s_{-i}),\, g^{-1}(i) > g'^{-1}(i).
    \end{equation}
    Note that only player $i\in N$ updates her strategy, so the following holds.
    \begin{equation}\label{eq:others_const}
        \forall\, g\in G(s_i, s_{-i}),\; \forall\, g'\in G(s_i', s_{-i}),\; k\le g'^{-1}(i)-1,\Rightarrow s_{g'(k)} = s_{g(k)} .
    \end{equation}

    We now show that $u_i(s_i', s_{-i}) > u_i(s_i, s_{-i})$.
    From Equation~\eqref{eq:waiting_time} and \eqref{eq:others_const}, for any $g'\in G(s_i', s_{-i})$, the waiting time of player $i$ can be written as
    \begin{equation}\label{eq:waiting_time_ver2}
        w_i^{g'}(s_i', s_{-i})
        =
        \max\Bigl\{\, \max_{1\le k\le g'^{-1}(i)-1} \{\, s_{g'(k)} + (g'^{-1}(i)-k)\gamma - s_i' \,\},\; 0 \Bigr\}.
    \end{equation}
    Hence, for all $g\in G(s_i, s_{-i})$ and all $g'\in G(s_i', s_{-i})$,
    \begin{align}
        s_i + &w_i^{g}(s_i, s_{-i}) = s_i + \max_{1\le k\le g^{-1}(i)}\{s_{g(k)} + (g^{-1}(i) - k)\gamma - s_i\} \notag\\
        &= \max_{1\le k\le g^{-1}(i)}\{s_{g(k)} + (g^{-1}(i) - k)\gamma \} \notag \\
        &= \max\Big\{\max_{1\le k\le g'^{-1}(i)-1}\{s_{g(k)} + (g^{-1}(i) - k)\gamma\}, \notag \\
        &\hspace{25mm}\max_{g'^{-1}(i)\le k\le g^{-1}(i)}\{s_{g(k)} + (g^{-1}(i) - k)\gamma\}\Big\} \notag \\
        &\quad (\because \text{Equation~\eqref{eq:minmax_g}})\notag\\
        &\ge \max\Big\{\max_{1\le k\le g'^{-1}(i)-1}\{s_{g(k)} + (g'^{-1}(i) - k)\gamma\},\notag \\
        &\hspace{25mm} \max_{g'^{-1}(i)\le k\le g^{-1}(i)}\{s_{g(k)} + (g^{-1}(i) - k)\gamma\}\Big\}\label{eq:uneq_1}\\
        &\quad (\because \text{Equation~\eqref{eq:minmax_g}})\notag\\
        &= \max\Big\{\max_{1\le k\le g'^{-1}(i)-1}\{s_{g'(k)} + (g'^{-1}(i) - k)\gamma\},\notag\\
        &\hspace{25mm} \max_{g'^{-1}(i)\le k\le g^{-1}(i)}\{s_{g(k)} + (g^{-1}(i) - k)\gamma\}\Big\} \notag\\
        &\quad (\because \text{Equation~\eqref{eq:others_const}})\notag\\
        &\ge \max\Big\{\max_{1\le k\le g'^{-1}(i)-1}\{s_{g'(k)} + (g'^{-1}(i) - k)\gamma\}, s_i\Big\}\notag\\
        &\quad (\because\text{Just taking $k=g^{-1}(i)$ in the second term.})\notag\\
        &\ge \max\Big\{\max_{1\le k\le g'^{-1}(i)-1}\{s_{g'(k)} + (g'^{-1}(i) - k)\gamma\}, s_i'\Big\}\quad (\because s_i'< s_i) \notag\\
        &= s_i' + \max\Big\{\max_{1\le k\le g'^{-1}(i)-1}\{s_{g'(k)} + (g'^{-1}(i) - k)\gamma - s_i'\}, 0\Big\} \notag\\
        &= s_i' + w_i^{g'}(s_i', s_{-i})\quad (\because \text{Equation~\eqref{eq:waiting_time_ver2}})\label{eq:ineq_service_time}
    \end{align}

    Note that equality in Equation~\eqref{eq:uneq_1} holds only when $g'^{-1}(i)=g^{-1}(i)$, which is the case covered by \eqref{eq:minmax_g} but not covered by \eqref{eq:noneq_g_g'}.
    Now we define a non-empty set 
    \begin{equation*}
        H:= \{(g, g')\in G(s_i, s_{-i})\times G(s_i', s_{-i})\mid g^{-1}(i)>g'^{-1}(i)\}.
    \end{equation*}
    Multiplying both sides of the Equation~\eqref{eq:ineq_service_time} by a positive constant, we obtain
    \begin{equation}
        \frac{s_i + w_i^{g}(s_i, s_{-i})}{|G(s_i, s_{-i})|\, |G(s_i', s_{-i})|} \ge \frac{s_i' + w_i^{g'}(s_i', s_{-i})}{|G(s_i, s_{-i})|\, |G(s_i', s_{-i})|}\quad (\forall (g, g')\notin H),\label{eq:uneq_2}
    \end{equation}
    \begin{equation}
        \frac{s_i + w_i^{g}(s_i, s_{-i})}{|G(s_i, s_{-i})|\, |G(s_i', s_{-i})|} > \frac{s_i' + w_i^{g'}(s_i', s_{-i})}{|G(s_i, s_{-i})|\, |G(s_i', s_{-i})|}\quad (\forall (g, g')\in H).\label{eq:uneq_3}
    \end{equation}
    Summing up Equation~\eqref{eq:uneq_2} and \eqref{eq:uneq_3} for all elements, we obtain
    \begin{align*}
        &\sum_{(g, g')\notin H} \frac{s_i + w_i^{g}(s_i, s_{-i})}{|G(s_i, s_{-i})|\, |G(s_i', s_{-i})|} + \sum_{(g, g')\in H} \frac{s_i + w_i^{g}(s_i, s_{-i})}{|G(s_i, s_{-i})|\, |G(s_i', s_{-i})|}\\
        &> \sum_{(g, g')\notin H}\frac{s_i' + w_i^{g'}(s_i', s_{-i})}{|G(s_i, s_{-i})|\, |G(s_i', s_{-i})|} + \sum_{(g, g')\in H}\frac{s_i' + w_i^{g'}(s_i', s_{-i})}{|G(s_i, s_{-i})|\, |G(s_i', s_{-i})|}
    \end{align*}
    \begin{equation*}
        \therefore \sum_{g\in G(s_i, s_{-i})} \frac{s_i + w_i^{g}(s_i, s_{-i})}{|G(s_i, s_{-i})|} > \sum_{g'\in G(s_i', s_{-i})}\frac{s_i' + w_i^{g'}(s_i', s_{-i})}{|G(s_i', s_{-i})|}
    \end{equation*}
    \begin{equation*}
        \therefore s_i + \mathbb{E}(w_i(s_i, s_{-i})) > s_i' + \mathbb{E}(w_i(s_i', s_{-i}))
    \end{equation*}
\end{proof}

\section{Proof of Lemma~\ref{lem:no-takeover}}\label{app3}

\begin{proof}
    This is shown by contradiction.
    Suppose $\exists i\in N, s_i^* \notin \{t_i\}\cup (S_i\setminus S_{i+1})$.
    Please note that this means that $s_i^*\neq t_i$ and $s_i^*\ge t_{i+1}$.
    Under this assumption, we will show that some player can improve his expected service order by deviating from $(s_1^*, \dots, s_n^*)$.
    
    In the case $s_{i+1}^* < s_i^*$, $\mathbb{E}\big(g_{(s_i^*, s_{-i}^*)}^{-1}(i+1)\big) < \mathbb{E}\big(g_{(s_i^*, s_{-i}^*)}^{-1}(i)\big)$ holds true.
    Then, player $i$ should change his strategy as $s_i=t_i$ because the following holds true:
    \begin{equation*}
        \mathbb{E}\big(g_{(t_i, s_{-i}^*)}^{-1}(i)\big) \le \mathbb{E}\big(g_{(s_i^*, s_{-i}^*)}^{-1}(i+1)\big) < \mathbb{E}\big(g_{(s_i^*, s_{-i}^*)}^{-1}(i)\big)\quad (\because t_i < s_{i+1}^*).
    \end{equation*}

    Similarly, in the case $s_{i+1}^* = s_i^*$, player $i$ should change his strategy as $s_i=t_i$ because the following holds true:
    \begin{equation*}
        \mathbb{E}\big(g_{(t_i, s_{-i}^*)}^{-1}(i)\big) < \mathbb{E}\big(g_{(s_i^*, s_{-i}^*)}^{-1}(i+1)\big) = \mathbb{E}\big(g_{(s_i^*, s_{-i}^*)}^{-1}(i)\big)\quad (\because t_i < s_i^*).
    \end{equation*}

    Finally, in the case $s_{i+1}^* > s_i^* \ge t_{i+1}$, player $i+1$ should change his strategy as $s_{i+1}=t_{i+1}$ because the following holds true:
    \begin{equation*}
        \mathbb{E}\big(g_{(t_{i+1}, s_{-(i+1)}^*)}^{-1}(i+1)\big) \le \mathbb{E}\big(g_{(s_{i+1}^*, s_{-(i+1)}^*)}^{-1}(i)\big) < \mathbb{E}\big(g_{(s_{i+1}^*, s_{-(i+1)}^*)}^{-1}(i+1)\big).
    \end{equation*}

    Based on the Lemma~\ref{lem:faster_better}, above discussion guarantee the statement.
\end{proof}

\section{Proof of Theorem~\ref{thm:nash_equilibria}}\label{app4}
\begin{proof}
    Firstly, consider the case $\exists i, j\in N,\; t_i=t_j$.
    Without losing generality, suppose $i < j$.
    According to Lemma~\ref{lem:no-takeover}, $s_i^*=t_i$.
    Suppose $s_j > s_i^* = t_j$ and let $l:=\max_{g\in G(s_j, s_{-j}^*)} g^{-1}(i)$.
    Because $s_i^* < s_j$, $l < \min_{g\in G(s_j, s_{-j}^*)}g^{-1}(j)$.
    \begin{equation*}
        \therefore l+1\le \mathbb{E}(g_{s_j, s_{-j}^*}^{-1}(j)).
    \end{equation*}
    Suppose $s_j^* = t_j = s_i^*$, then $\max_{g\in G(s_j^*, s_{-j}^*)}g^{-1}(j)=l+1$ and $\min_{g\in G(s_j^*, s_{-j}^*)}g^{-1}(j)\le l$.
    \begin{equation*}
        \therefore \mathbb{E}(g_{s_j^*, s_{-j}^*}^{-1}(j)) < l+1 \le \mathbb{E}(g_{s_j, s_{-j}^*}^{-1}(j)).
    \end{equation*}
    Thus, from Lemma~\ref{lem:faster_better}, we obtain $s_j^* = t_j$.

    Next, suppose $t_1< t_2$.
    In this case, from Lemma~\ref{lem:no-takeover}, $s_1^*\in [t_1, t_2)$, which means that the player $1$ arrives first; $\forall g\in G(s^*),\; g^{-1}(1)=1$.
    Because there exists no queue in front of the first player, going as fast as possible maximises his utility.
    Thus, we obtain $s_1^* = t_1$.

    Finally, we consider a player $i\in N$ such that $t_{i-1} < t_i < t_{i+1}$.
    From Lemma~\ref{lem:no-takeover}, $\forall j\in\{1, \dots, i-1\},\; s_j^*<t_i$ and $s_i^*<t_{i+1}$ hold true.
    Thus, $\forall g\in G(s_i^*, s_{-i}^*),\; g^{-1}(i)=i$.
    \begin{align}
        \therefore \forall g\in G(s_i^*, s_{-i}^*),\, w_i^g(s_i^*, s_{-i}^*) &= \max_{1\le k\le g^{-1}(i)}\{s_{g(k)}^* + (g^{-1}(i) - k)\gamma - s_i^*\} \notag \\
        &= \max_{1\le k\le i}\{s_{g(k)}^* + (i - k)\gamma - s_i^*\} \quad  \label{eq:waiting_const}
    \end{align}

    Based on Lemma~\ref{lem:same_arrival_time}, Equation~\eqref{eq:waiting_const} means that player $i$'s waiting time $w_i^g(s_i^*, s_{-i}^*)$ does not depend on the service order $g\in G(s_i^*, s_{-i}^*)$.
    Therefore, we focus on a specific service order $g^*\in G(s_i^*, s_{-i}^*)$ such that $\forall k\in \{1, \dots, n\}, g^*(k) = k\in N$.
    Herein, $g^*\in G(s_i^*, s_{-i}^*)$ because $s_1^*\le \dots \le s_n^*$ is satisfied from Lemma~\ref{lem:no-takeover}.
    \begin{align}
        \therefore \mathbb{E}(w_i(s_i^*, s_{-i}^*)) &= \sum_{g\in G(s_i^*, s_{-i}^*)} \frac{1}{|G(s_i^*, s_{-i}^*)|} w_i^g(s_i^*, s_{-i}^*) \notag \\
        &= w_i^{g^*}(s_i^*, s_{-i}^*) \notag \\
        &= \max_{1\le k\le i}\{s_{g^*(k)}^* + (g^{*-1}(i) - k)\gamma - s_i^*\} \notag \\
        &= \max_{1\le k\le i}\left\{s_k^* + (i - k)\gamma - s_i^*\right\}
    \end{align}
    \begin{align}
        \therefore s_i^* + \mathbb{E}(w_i(s_i^*, s_{-i}^*)) &= \max_{1\le k\le i}\{s_k^* + (i - k)\gamma\} \notag\\
        &= \max\left\{\max_{1\le k\le i-1}\{s_k^* + (i - k)\gamma\}, s_i^*\right\}
    \end{align}
    Let $\theta_i := \max_{1\le k\le i-1}\{s_k^* + (i - k)\gamma\}$ which is independent to $s_i^*$.
    In the case $\theta_i \ge t_i$,
    \begin{equation}\label{eq:service_time}
        \max\{\theta_i, s_i^*\} = 
        \begin{cases}
            \theta_i\quad &(\text{if $t_i\le s_i^*\le \theta_i$}),\\
            s_i^*\quad &(\text{if $\theta_i < s_i^*$}).
        \end{cases}
    \end{equation}
    Thus, with Lemma~\ref{lem:no-takeover}, player $i$'s best response is described as $[t_i, \theta_i]\cap [t_i, t_{i+1})$.
    In the case $\theta < t_i$, $\max\{\theta_i, s_i^*\} = s_i^*$ because $s_i^* \ge t_i$, which leads $s_i^* = t_i$.

    Now we will show that an arbitrary strategy profile $s^* := (s_1^*, \dots, s_n^*) \in \prod_{i\in N}\Theta_i$ is a Nash equilibrium.
    In the case $i\neq j\in N,\; t_i=t_j$, $s_i^* = s_j^* = t_i$.
    If player $i$ changes his strategy as $s_i' > t_i$,
    \begin{equation*}
        \min_{g\in G(s_i', s_{-i}^*)} g^{-1}(i) = \max_{g\in G(s^*)} g^{-1}(i) =: l'.
    \end{equation*}
    \begin{equation*}
        \therefore \mathbb{E}(g^{-1}(s_i', s_{-i}^*)(i)) \ge l' > \mathbb{E}(g^{-1}(s^*)(i)).
    \end{equation*}
    Thus, from Lemma~\ref{lem:faster_better}, this change is not profitable for player $i$.

    We finally consider the case $t_{i-1}<t_i<t_{i+1}$.
    If player $i$'s updated strategy is $s_i'\in \Theta_i$, 
    \begin{align*}
        s_i' + \mathbb{E}(w_i(s_i', s_{-i}^*)) &= \theta_i = s_i^* + \mathbb{E}(w_i(s_i^*, s_{-i}^*))\quad (\because \text{Equation~\eqref{eq:service_time}}).
    \end{align*}
    Thus, this kind of update cannot improve player $i$'s utility.

    If player $i$'s updated strategy is $s_i'\notin \Theta_i$, $s_i'>s_i^*$ holds true.
    If $s_i' \ge s_{i+1}^*$, $\exists g\in G(s_i', s_{-i}^*),\; g^{-1}(i) \ge i+1$.
    \begin{equation}
        \therefore \mathbb{E}(g^{-1}(s_i', s_{-i}^*)(i)) > i = \mathbb{E}(g^{-1}(s^*)(i)).
    \end{equation}
    Thus again, from Lemma~\ref{lem:faster_better}, this change is not profitable for player $i$.

    If $s_i' < s_{i+1}^*$, $\forall g\in G(s_i', s_{-i}^*), g^{-1}(i)=i$.
    Therefore, $G(s_i', s_{-i}^*) = G(s^*)$ holds true.
    \begin{align*}
        \therefore \forall& g\in G(s^*),\; w_i^g(s^*) = \max_{1\le k\le g^{-1}(i)}\big\{\,s_{g(k)}^* + \big(g^{-1}(i) - k\big)\gamma - s_i^*\,\big\}\\
        &= \max\left\{\max_{1\le k\le i-1}\big\{\,s_{g(k)}^* + \big(i - k\big)\gamma - s_i^*\,\big\}, 0\right\}\\
        &= \max\left\{\max_{1\le k\le i-1}\big\{\,s_{g(k)}^* + \big(i - k\big)\gamma - s_i'\,\big\}+s_i' - s_i^*, 0\right\}\\
        &\le \max\left\{\max\left\{\max_{1\le k\le i-1}\big\{\,s_{g(k)}^* + \big(i - k\big)\gamma - s_i'\,\big\},\; 0\right\}+s_i' - s_i^*, 0\right\}\\
        &= \max\left\{w_i^g(s_i', s_{-i}^*)+s_i' - s_i^*, 0\right\}\\
        &= w_i^g(s_i', s_{-i}^*)+s_i' - s_i^*\quad (\because s_i'>s_i^*).
    \end{align*}
    Thus, $s_i^* + w_i^g(s^*) \le s_i' + w_i^g(s_i', s_{-i}^*)$ holds true, which means that $s_i'$ does not improve player $i$'s utility.
\end{proof}

\section{Proof of Theorem~\ref{thm:bayesian_nash_equilibrium}}\label{app5}

\begin{proof}
    First, we show that $\pi_i^*: t_i \mapsto t_i$ for all $i \in N$ constitutes a Bayesian Nash equilibrium.
    Fix a player $i \in N$ and an arbitrary admissible strategy $\pi_i : T \to S_i$.
    By feasibility of arrival times we have $\pi_i(t_i) \ge t_i$ for all $t_i \in T$, and $\pi_i^*(t_i) = t_i$ for all $t_i \in T$.

    Consider an arbitrary type profile $t = (t_i,t_{-i}) \in T^n$.
    If $\pi_i(t_i) = t_i$, then player $i$ takes the same action under $\pi_i$ and $\pi_i^*$, and hence
    \begin{equation*}
        u_i\bigl(\pi_i^*(t_i), \pi_{-i}^*(t_{-i})\bigr)
        =
        u_i\bigl(\pi_i(t_i), \pi_{-i}^*(t_{-i})\bigr).
    \end{equation*}
    If $\pi_i(t_i) > t_i$, then by Lemma~\ref{lem:faster_better}, going faster never makes player $i$ worse off given the others' strategies and strictly improves his expected utility whenever his expected position in the berth allocation is strictly improved.
    Since $\pi_i^*(t_i) = t_i < \pi_i(t_i)$, Lemma~\ref{lem:faster_better} implies
    \begin{equation*}
        u_i\bigl(\pi_i^*(t_i), \pi_{-i}^*(t_{-i})\bigr) \ge u_i\bigl(\pi_i(t_i), \pi_{-i}^*(t_{-i})\bigr)
    \end{equation*}
    for every $t \in T^n$, with strict inequality whenever player $i$'s expected berth position is strictly improved by switching from $\pi_i(t_i)$ to $t_i$.
    \begin{equation}\label{eq:app5_pointwise}
        \therefore \forall t\in T^n,\;u_i\bigl(\pi_i^*(t_i), \pi_{-i}^*(t_{-i})\bigr)
        \;\ge\;
        u_i\bigl(\pi_i(t_i), \pi_{-i}^*(t_{-i})\bigr).
    \end{equation}

    Let $P$ be the common prior on $T^n$ with density $p$ with respect to the Lebesgue measure, as defined in Section~\ref{subsec:incomplete-info-game}.
    Integrating inequality~\eqref{eq:app5_pointwise} with respect to $P$ yields
    \begin{equation}\label{eq:app5_BNE_weak}
        \int_{T^n} u_i\bigl(\pi_i^*(t_i), \pi_{-i}^*(t_{-i})\bigr)\,p(t)\,dt
        \;\ge\;
        \int_{T^n} u_i\bigl(\pi_i(t_i), \pi_{-i}^*(t_{-i})\bigr)\,p(t)\,dt.
    \end{equation}
    Hence $(\pi_i^*)_{i \in N}$ is a Bayesian Nash equilibrium in the sense defined by Equation~\ref{eq:def-bne}.

    Next, we show uniqueness among symmetric Bayesian Nash equilibria.
    Suppose that there exists a symmetric Bayesian Nash equilibrium $(\pi_i^s)_{i \in N}$ with $\pi_i^s = \pi_j^s$ for all $i,j \in N$, which is not equivalent to $(\pi_i^*)_{i \in N}$.
    Then the common strategy $\pi^s$ must differ from the identity map $t \mapsto t$ on a set of types with positive prior probability.
    Define
    \begin{equation*}
        D := \{ t \in T : \pi^s(t) > t \}.
    \end{equation*}
    And let $P_i$ denote the marginal distribution of $t_i$ under $P$, that is,
    \begin{equation*}
        P_i(B) := P\bigl(\{t \in T^n : t_i \in B\}\bigr)\quad \text{for all Borel sets } B \subseteq T.
    \end{equation*}

    If $P_i(D) = 0$ for the $i$-th marginal of $P$, then $\pi^s(t) = t$ for $P_i$-almost all $t$, that is, $(\pi_i^s)_{i \in N}$ is equivalent to $(\pi_i^*)_{i \in N}$, which contradicts the assumption that they are different.
    Hence $P_i(D) > 0$.
    Since $\pi^s(t) - t$ is nonnegative on $D$, there exists an integer $k \ge 1$ such that the set
    \begin{equation*}
        D_k := \{ t \in T : \pi^s(t) \ge t + 1/k \}
    \end{equation*}
    satisfies $P_i(D_k) > 0$.
    For each $t_i \in D_k$, the deviation from $\pi^s(t_i)$ to $t_i$ strictly reduces his arrival time by at least $1/k$.

    Consider player $i$ and define a deviation $\pi_i^\dagger : T \to S_i$ by
    \begin{equation*}
        \pi_i^\dagger(t_i)
        =
        \begin{cases}
            t_i & \text{if } t_i \in D_k,\\
            \pi^s(t_i) & \text{otherwise}.
        \end{cases}
    \end{equation*}
    Other players $j \ne i$ keep using $\pi_j^s = \pi^s$.
    For any type profile $t = (t_i,t_{-i})$ with $t_i \in D_k$, the deviation from $\pi^s(t_i)$ to $t_i$ weakly improves player $i$'s berth position and strictly improves it whenever there exists another player $j$ whose berth position lies between the two arrival times.
    Because the common prior has a strictly positive density on $T^n$, and because $P_i(D_k) > 0$, the set of type profiles for which $t_i \in D_k$ and at least one such $j$ exists has strictly positive prior probability.
    \footnote{
    Since $P$ has a density $p$ that is strictly positive almost everywhere, the marginal distribution of $t_i$ admits a density that is strictly positive almost everywhere, and the set $D_k$ must have positive Lebesgue measure.
    For any $t_i \in D_k$, the interval $(t_i, t_i + 1/k)$ is contained in $(t_i,\pi^s(t_i))$, and by Fubini's theorem the set
    \begin{equation*}
        E := \{ t \in T^n : t_i \in D_k,\ \exists j \ne i \text{ with } t_i < t_j < t_i + 1/k \}
    \end{equation*}
    has positive Lebesgue measure.
    Since $p(t) > 0$ for almost all $t \in T^n$, it follows that $P(E) = \int_E p(t)\,dt > 0$.}
    By Lemma~\ref{lem:faster_better}, player $i$'s utility strictly increases on this set and never decreases.
    Therefore,
    \begin{equation*}
        \int_{T^n} u_i\bigl(\pi_i^\dagger(t_i), \pi_{-i}^s(t_{-i})\bigr)\,p(t)\,dt
        \;>\;
        \int_{T^n} u_i\bigl(\pi_i^s(t_i), \pi_{-i}^s(t_{-i})\bigr)\,p(t)\,dt,
    \end{equation*}
    which contradicts the assumption that $(\pi_i^s)_{i \in N}$ is a Bayesian Nash equilibrium.
    Hence any symmetric Bayesian Nash equilibrium must coincide with $(\pi_i^*)_{i \in N}$ almost surely under $P$.

    % Collecting these arguments, we conclude that $(\pi_i^*)_{i \in N}$ is a Bayesian Nash equilibrium and that any symmetric Bayesian Nash equilibrium must coincide with $(\pi_i^*)_{i \in N}$ up to changes on sets of types with prior probability zero.
    % This proves Theorem~\ref{thm:bayesian_nash_equilibrium}.
\end{proof}

\end{document}